\documentclass[lettersize,journal]{IEEEtran}
\usepackage{amsmath,amsfonts}
\usepackage{algorithmic}
\usepackage{algorithm}
\usepackage{array}
\usepackage[caption=false,font=normalsize,labelfont=sf,textfont=sf]{subfig}
\usepackage{textcomp}
\usepackage{stfloats}
\usepackage{url}
\usepackage{verbatim}
\usepackage{graphicx}
\usepackage{cite}

\usepackage{multirow}
\usepackage{makecell}   
\usepackage{booktabs}
\usepackage{dsfont}     
\usepackage{amssymb}    
\usepackage{amsthm}
\usepackage{tabularray}
\usepackage[colorinlistoftodos]{todonotes}
\usepackage{soul}  

\newtheorem{definition}{Definition}
\newtheorem{theorem}{Theorem}

\usepackage[colorlinks,
            linkcolor=blue,
            anchorcolor=blue,
            citecolor=blue]{hyperref}

\begin{document}

\title{Trustworthy Neighborhoods Mining: Homophily-Aware Neutral Contrastive Learning for Graph Clustering}

\author{Liang Peng, Yixuan Ye, Cheng Liu$^{*}$,
    Hangjun Che,
    Man-Fai Leung,
	Si Wu,
	and~Hau-San Wong
\thanks{Liang Peng and Yixuan Ye are with the Department of Computer Science, Shantou University (email: 23lpeng@stu.edu.cn; 22yxye2@stu.edu.cn)}
\thanks{Cheng Liu is with the College of Computer Science and Technology, Huaqiao University and the Department of Computer Science, Shantou University (email:chengliu10@gmail.com)}
	\thanks{Hangjun Che is with the Chongqing Key Laboratory of Nonlinear Circuits and Intelligent Information Processing, College of Electronic and Information Engineering, Southwest University, Chongqing, China (email: hjche123@swu.edu.cn)}
	\thanks{Man-Fai Leung is with the School of Computing and Information Science, Faculty of Science and Engineering, Anglia Ruskin University, Cambridge, U.K. (e-mail: man-fai.leung@aru.ac.uk)}
    \thanks{Si Wu is with the School of Computer Science and Engineering, South China University of Technology. (email: cswusi@scut.edu.cn)}
	\thanks{Hau-San Wong is with the Department of Computer Science, City University of Hong Kong. (email: cshswong@cityu.edu.hk)}
	\thanks{$^{*}$Corresponding author: Cheng Liu}

}



\maketitle

\begin{abstract}
Recently, neighbor-based contrastive learning has been introduced to effectively exploit neighborhood information for clustering. However, these methods rely on the homophily assumption—that connected nodes share similar class labels and should therefore be close in feature space—which fails to account for the varying homophily levels in real-world graphs.
As a result, applying contrastive learning to low-homophily graphs may lead to indistinguishable node representations due to unreliable neighborhood information, making it challenging to identify trustworthy neighborhoods with varying homophily levels in graph clustering.
To tackle this, we introduce a novel neighborhood Neutral Contrastive Graph Clustering method NeuCGC that extends traditional contrastive learning by incorporating neutral pairs—node pairs treated as weighted positive pairs, rather than strictly positive or negative. These neutral pairs are dynamically adjusted based on the graph’s homophily level, enabling a more flexible and robust learning process. 
Leveraging neutral pairs in contrastive learning, our method incorporates two key components: 1) an adaptive contrastive neighborhood distribution alignment that adjusts based on the homophily level of the given attribute graph, ensuring effective alignment of neighborhood distributions, and 2) a contrastive neighborhood node feature consistency learning mechanism that leverages reliable neighborhood information from high-confidence graphs to learn robust node representations, mitigating the adverse effects of varying homophily levels and effectively exploiting highly trustworthy neighborhood information.
Experimental results demonstrate the effectiveness and robustness of our approach, outperforming other state-of-the-art graph clustering methods. Our code is available at \url{https://github.com/THPengL/NeuCGC}.
\end{abstract}


\section{Introduction}
\IEEEPARstart{R}{eal-world} applications frequently involve various types of graph data, which offer valuable insights into the relationships between elements. Graph clustering is a fundamental task in graph data analysis, aiming to group nodes based on learned features.
Recently, various methods have been developed, achieving significant advancements in graph clustering by employing diverse strategies \cite{2017MGAE-wang},\cite{2020AGE-cui},\cite{2020CGCN-hui},\cite{2020MVGRL-hassani}.
Among these methods, the contrastive learning paradigm \cite{2006CL-hadsell, 2018CPC-oord} has emerged as a promising approach, leveraging different augmented views of the graph as an effective self-supervised technique \cite{2020MVGRL-hassani},\cite{2020GraphCL-you},\cite{2021BGRL-R2-thakoor}, \cite{2022BGRL-thakoor-Large},\cite{2022AutoGCL-yin},\cite{2023ABGML-R2-chen},\cite{2024TopoGCL-chen}. The main idea behind traditional contrastive graph clustering (CGC) methods is to learn robust representations by contrasting positive pairs with negative pairs derived from these augmented views \cite{2022AGC-DRR-Gong, 2023CONVERT-yang}.
Additionally, several studies have explored leveraging the inherent information in a node's neighborhood for neighbor contrastive learning, leading to notable improvements \cite{2021GraphMLP-hu-graph},\cite{2022LocalGCL-zhang-localized},\cite{2023SCGC-Liu},\cite{2023NCLA-shen-neighbor},\cite{11209498},\cite{xian2025mdhgfn}.

However, a critical assumption underlying most existing graph neighborhood contrastive learning methods is the homophilic nature of graph data, where connected nodes are typically expected to belong to the same class and be close to each other in the learned representation space. To some extent, it ensures the reliability of information in certain neighborhoods. 
While in practice, real-world graphs can exhibit varying levels of homophily \cite{2021CPGNN-zhu}, \cite{2022Hete-zheng}, and those with low homophily can result in neighbors having different labels, which distorts the essential neighborhood information used for contrastive learning.
To address these challenges, some graph learning methods have attempted to leverage graph homophily to learn reliable relationships between neighbors from different perspectives. For instance, \cite{2022ProGCL-xia} addresses the issue of false negatives by refining the criteria for selecting hard negatives. Additionally, \cite{2023DGCN-pan-beyond} introduces a mixed graph filter designed to capture both low-frequency and high-frequency information within the graph, thereby making the method adaptable to varying levels of homophily. 
Nevertheless, existing methods still face two key challenges: 
\textbf{\textit{(1) How to directly leverage graph homophily, particularly in correctly identifying existing or newly incorporated  neighboring nodes as reliable positive sample pairs for contrastive learning, and (2) How to mitigate the negative impact of low homophily in graphs, and effectively exploit highly reliable neighborhood information for embedding learning, require further exploration.}}

To address these challenges, we propose a novel homophily-aware Neutral Contrastive Graph Clustering approach \textbf{NeuCGC} for mining trustworthy neighborhoods in contrastive learning. This method extends traditional neighborhood contrastive learning by incorporating \textit{neutral sample pairs}—pairs of nodes treated as weighted positive pairs instead of strictly positive or negative. These neutral pairs are dynamically adjusted based on the graph's homophily level, allowing for a more flexible and robust learning process.
By leveraging neutral pairs in contrastive learning, our method incorporates two key components:
Firstly, neutral pairs are adaptively selected as positive pairs for contrastive learning based on the estimated level of graph homophily. In other words, in graphs with low homophily, neutral pairs among neighbor nodes will contribute less to contrastive learning, while in graphs with high homophily, they will contribute more. Building upon this, we propose a \textit{neutral contrastive learning (NeuCL)} mechanism that minimizes the distance of positive pairs, partially minimizes that of neutral pairs, while maximizes that of negative pairs. To leverage this, we propose a neighborhood distribution neutral contrastive alignment loss that incorporates these neutral pairs using a homophily estimation strategy. This dynamic contrastive distribution alignment enables the weighted adjustment of neutral pairs as positive pairs, ensuring better adaptation to graphs with varying homophily levels and improving the robustness of contrastive learning.
Secondly, in addition to the attribute graph, we further leverage reliable neighborhood information learned from highly confident graphs based on node representations. We propose an adaptive feature consistency neutral contrastive learning module that utilizes this reliable neighborhood information to mitigate the negative impact of varying homophily levels, while effectively exploiting trustworthy neighborhood relationships to enhance node representation learning.

Our primary contributions can be summarized as follows:

\begin{itemize}
    \item We introduce neutral pairs to mine trustworthy neighborhood information in graph contrastive learning, ensuring more reliable learning across graphs with varying homophily levels. 

    \item By leveraging neutral pairs to capture reliable neighborhood information, we propose NeuCGC comprising a neutral contrastive alignment mechanism and adaptive feature consistency learning, which enhance the robustness and flexibility of contrastive graph clustering.

    \item Extensive experiments on homophilic and heterophilic datasets empirically demonstrate that our approach outperforms other baselines, particularly enjoying a performance improvement of 13\% on average in terms of ARI.
\end{itemize}

\section{Related Work}
\subsection{Deep Attributed Graph Clustering} 
As outlined in recent survey \cite{2022survey-liu}, in addition to contrastive graph clustering, most existing deep attributed graph clustering (DGC) techniques fall into two main categories: reconstructive approaches and adversarial approaches. Early study such as MGAE \cite{2017MGAE-wang} employs graph autoencoders \cite{2016VGAE-kipf} to reconstruct node features for representation learning. Other works, including AGAE \cite{2019AGAE-Tao}, ARGA, and ARVGA \cite{2020ARGA-Pan}, adopt adversarial strategies to generate self-supervised signals that encourage more discriminative embeddings. Recent DGC methods have significantly advanced clustering performance by integrating structural and semantic information. For instance, SDCN \cite{2020SDCN-bo} and DFCN \cite{2021DFCN-tu-deep} fuse autoencoders with graph neural networks (GNNs) \cite{2016GCN-kipf} to jointly capture attribute and structural features. More recently, AutoSSL \cite{2022AutoSSL-jin} and DyFSS \cite{2024DyFSS-zhu} incorporate multiple self-supervised tasks to improve clustering quality. To mitigate feature over-smoothing, AGC and IAGC \cite{2023AGC-Zhang} adopt a $k$-order graph convolution to capture global cluster structures. Additionally, SDAC-DA \cite{2024SDAC-DA-Berahmand} introduces pairwise constraint selection for semi-supervised graph clustering.

\subsection{Contrastive Graph Clustering}
A well-designed self-supervised objective is crucial for deep clustering, and contrastive learning \cite{2006CL-hadsell} has become a powerful paradigm in graph clustering \cite{2020MVGRL-hassani}. The core idea of graph contrastive learning (GCL) is to construct positive pairs—typically through graph augmentations—and learn representations by maximizing consistency between them while contrasting against negative ones \cite{2019DGI-velickovic}. 
For instance, GraphCL \cite{2020GraphCL-you} leverages predefined graph augmentations and the InfoNCE \cite{2018CPC-oord} loss to maximize mutual information consistency between positive pairs, thereby learning invariant features. 
Similarly, traditional CGC methods such as AGC-DRR \cite{2022AGC-DRR-Gong}, DCRN \cite{2022DCRN-liu}, and CONVERT \cite{2023CONVERT-yang} design diverse augmentation strategies to generate positive and negative pairs for contrastive learning. However, they consider only distinct views of the same node as positive pairs, disregarding the informative neighborhood structure inherent in graphs. 
To address this limitation, neighbor-based CGC approaches like SCGC \cite{2023SCGC-Liu}, NCLA \cite{2023NCLA-shen-neighbor}, LocalGCL \cite{2022LocalGCL-zhang-localized}, and NS4GC \cite{2024NS4GC-Liu} incorporate neighborhood information by treating a node and its neighbors as positive pairs. Nevertheless, these methods heavily rely on the homophily assumption, which may not hold in heterophilic graphs prevalent in many real-world scenarios \cite{2022DSSL-xiao-decoupled} \cite{2024R4-li-permutation}.

\subsection{Graph Clustering under Varying Homophily Levels.}
Homophily, a fundamental property of graph data, describes the tendency of nodes within the same class to be more likely connected. Real-world graphs exhibit diverse degrees of homophily, and recent methods leveraging this property have achieved promising performance. For instance, HomoGCL \cite{2023HomoGCL-Li} adopts node saliency to estimate the probability of random augmented neighbors being positive for graph contrastive learning. DGCN \cite{2023DGCN-pan-beyond} and AHGFC \cite{2024AHGFC-wen-homophily} introduce graph filters to capture both homophilic and heterophilic structure information for graph clustering or muti-view graph clustering. Recently, HeterGCL \cite{2024HeterGCL-Wang} combines structural and semantic cues to enhance representation learning across graphs with different homophily ratios. Despite these advances, effectively leveraging graph homophily for contrastive learning, especially for correctly identifying existing neighbors or discovering more semantically similar neighbors as reliable positive pairs, remains an under-explored challenge. In contrast to these methods, as well as traditional CGC and neighbor-base CGC ones, our approach addresses this gap by introducing a neutral contrastive alignment technique along with an adaptive feature consistency neutral contrastive learning module. This design enables the reliable exploitation of neighborhood structure information without relying on the homophily assumption, and thus applicable to real-world graphs with diverse levels of homophily.
\begin{figure*}[!t]
    \centering{
        \includegraphics[width=1\textwidth]{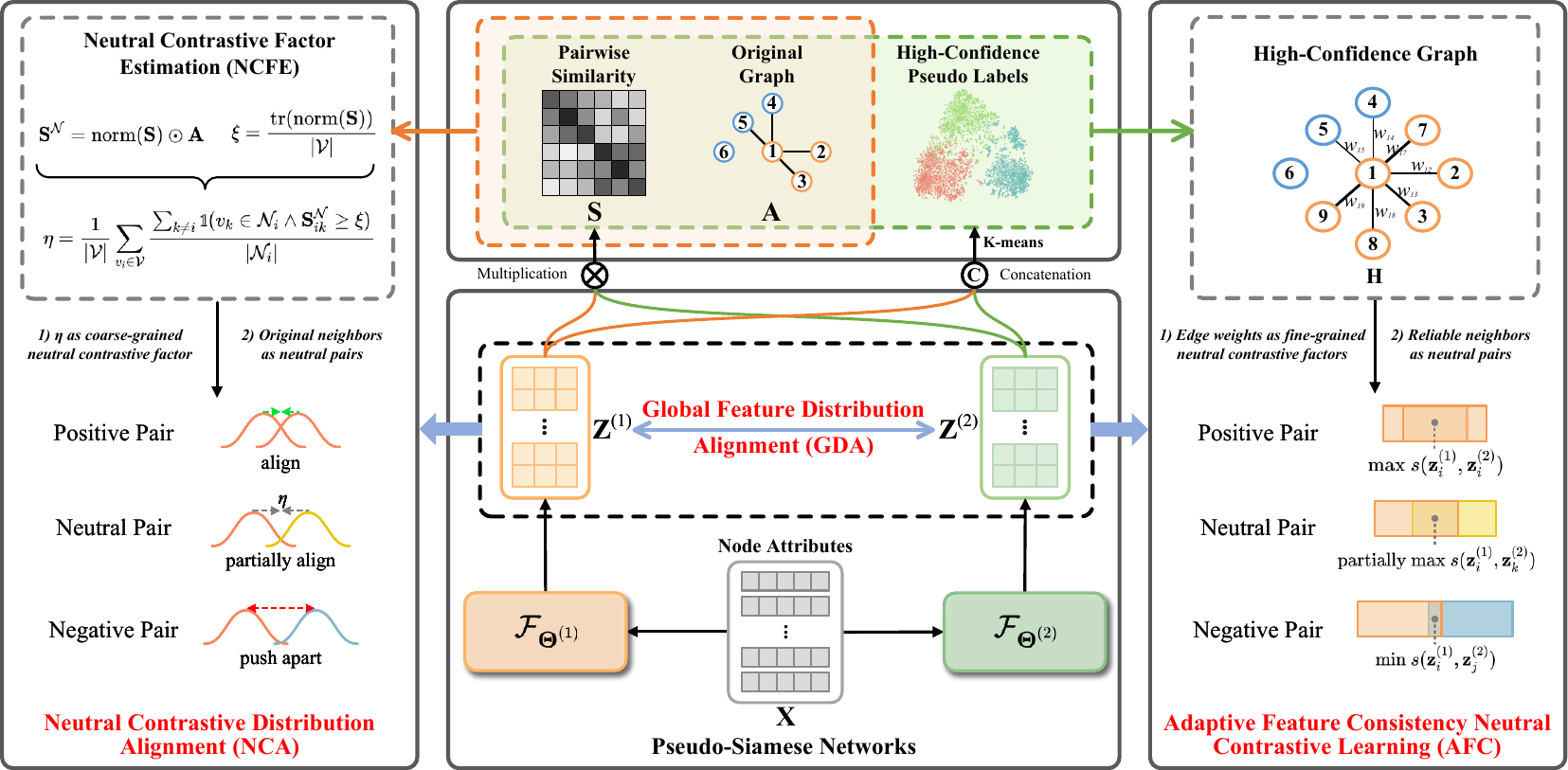}
        \caption{Overview of the proposed NeuCGC. 
        Given the node attributes $\textbf{X}$ and the adjacency matrix $\textbf{A}$ of a graph, we employ a Pseudo-Siamese Networks to encode $\textbf{X}$ into embeddings $\textbf{Z}^{(1)}$ and $\textbf{Z}^{(2)}$, and jointly optimize the Global Feature Distribution Alignment (GDA), Neutral Contrastive Distribution Alignment (NCA), and Adaptive Feature Consistency Neutral Contrastive Learning (AFC). Specifically, GDA facilitates global information sharing. NCA adopts the Neutral Contrastive Factor Estimation (NCFE) technique which leverages $\textbf{A}$ and node pairwise similarity $\textbf{S}$ to estimate a coarse-grained neutral contrastive factor $\eta$, thereby enabling neutral contrastive learning and enhancing local neighborhood information interaction. AFC employs $\textbf{S}$, $\textbf{A}$, and reliable pseudo-labels to construct a high-confidence graph $\textbf{H}$, which is subsequently used to enhance the feature consistency $s(\cdot)$ of node embeddings.}
    \label{fig_overview}
    }
\end{figure*}

\section{Preliminaries}
\subsection{Notations and Problem Definition}
Given an attributed graph \(\mathcal{G}=(\mathcal{V},\mathcal{E}, \textbf{X})\), where \(\mathcal{V} = \{v_1, v_2, \ldots, v_N\}\) is the vertex set with $N$ nodes and \(|\mathcal{V}|=N\), \(\mathcal{E}\) is the edge set, and \(\textbf{X} = [\textbf{x}_1, \textbf{x}_2, \ldots, \textbf{x}_N]^{\top} \in \mathbb{R}^{N \times D}\) is the corresponding node attribute matrix (with \(\textbf{x}_i \in \mathbb{R}^{D}\) as the attribute of node \(v_i\) and \(D\) as the number of features). Let \(\textbf{A} \in [0,1]^{N \times N}\) be the adjacency matrix where \(\textbf{A}_{ij} = 1\) if edge \((i,j) \in \mathcal{E}\) and \(0\) otherwise. 

Consider a deep graph clustering task where no ground truth labels are provided during training; the goal is to cluster similar nodes in $\mathcal{V}$ of $\mathcal{G}$ based on the learned features $\textbf{Z}$ from their attributes $\textbf{X}$ by neural network  in an unsupervised manner.
In general, the input graph $\mathcal{G}$ is feed to a neural network $f_{\mathbf{\Theta}}$ parameterized by $\mathbf{\Theta}$, such that the learned node embeddings $\textbf{Z}=f_{\mathbf{\Theta}} (\mathcal{G})= f_{\mathbf{\Theta}} (\textbf{X}, \textbf{A}) \in \mathbb{R}^{N\times d}$ in latent space can be used to divide nodes into disjoint clusters with the application of clustering algorithms such as K-means or spectral clustering, where $d$ is the dimension of latent features. A summary of these notations is provided in the Supplementary Material.

\subsection{Homophily Ratio and Neighborhood Homophily Ratio}
The homophily ratio is a quantitative measure that helps to quantify the extent of homophily in a graph $\mathcal{G}$.
\begin{definition}[Homophily Ratio \cite{2020rh_Zhu}]\label{df_r_h}
The homophily ratio $r_h$ in a graph $\mathcal{G}=  (\mathcal{V},\mathcal{E})$ with node set $\mathcal{V}$ and edge set $\mathcal{E}$ is the fraction of edges connecting nodes within the same class to the total number of edges: 
\begin{equation}\label{eq_r_h}
    r_h = \frac{|\{(i, j) : (i, j) \in \mathcal{E} \wedge y_i = y_j\}|}{|\mathcal{E}|},
\end{equation}
where $(i, j)$ is the edge connecting nodes $v_i$ and $v_j$, $y_i$ and $y_j$ are the ground-truth labels of them, respectively. 
\end{definition}

When focusing on the homophily within node neighborhoods, we can derive another measure to quantify neighborhood homophily.
\begin{definition}[Neighborhood Homophily Ratio \cite{2020rnh-mpei}]\label{df_r_nh}
The neighborhood homophily ratio $r_{nh}$ is the average homophily ratio across all node neighborhoods.
\begin{equation}\label{eq_r_nh}
    r_{nh} = \frac{1}{|\mathcal{V}|} \sum_{v_i \in \mathcal{V}} \frac{|\{v_j : v_j \in \mathcal{N}_i \wedge y_i = y_j\}|}{|\mathcal{N}_i|},
\end{equation}
where $\mathcal{N}_i = \{v_k: v_k \in \mathcal{V} \wedge (i,k) \in \mathcal{E}\}$ denotes the neighbor set of $v_i$, and $|\mathcal{N}_i|$ is the number of $v_i$'s neighbors. 
\end{definition}

The neighborhood homophily ratio $r_{nh}$ of $\mathcal{G}$ represents the average homophily ratio across all node neighborhoods. Compared to homophily ratio $r_{h}$, it provides insight into the average homophily extent at node level.
Table \ref{tb_Statistics} provides detailed statistics of real-world graph datasets, where we can observe varying levels of homophily ratio $r_h$ and neighborhood homophily ratio $r_{nh}$. In general, an attributed graph can be considered homophilic if it has $r_h \rightarrow 1$, and heterophilic if $r_h \rightarrow 0$ \cite{2020rh_Zhu},\cite{2022KDGA-wu-knowledge},\cite{2023GraphACL-Xiao}. Therefore, effectively leveraging varying levels of graph homophily to accurately identify reliable neighborhood information is crucial for neighborhood contrastive learning in real-world applications.

\begin{table}[!t]
    \caption{Statistics of homophilic and heterophilic graph datasets.}
    \begin{center}
    \setlength{\tabcolsep}{1.0mm}{
    \begin{tabular}{cccccccc}
    \toprule
    \textbf{Dataset} &\textbf{\#Nodes} &\textbf{\#Edges} &\textbf{\#Classes} &\textbf{\#Attributes} &\(r_h\) &\(r_{nh}\) &\(\delta\) \\
    \hline 
    \textbf{Cora} &2708 &5278 &7 &1433 &0.81 &0.83 &0.0088 \\
    \textbf{Citeseer} &3327 &4552 &6 &3703 &0.74 &0.72 &0.0042 \\
    \textbf{ACM} &3025 &13128 &3 &1870 &0.82 &0.88 &0.0089 \\
    \textbf{DBLP} &4057 &3528 &4 &334 &0.80 &0.84 &0.0023 \\
    \textbf{Photo} &7650 &119081 &8 &745 &0.83 &0.85 &0.0266 \\
    \textbf{Pubmed} &19717 &44324 &3 &500 &0.80 &0.79 &0.0006 \\
    \textbf{Texas} &183 &309 &5 &1703 &0.06 &0.06 &0.0341 \\
    \textbf{Wisconsin} &251 &466 &5 &1703 &0.18 &0.16 &0.0371 \\
    \textbf{Cornell} &183 &295 &5 &1703 &0.30 &0.30 &0.0303 \\
    \textbf{Chameleon} &2277 &36101 &5 &2325 &0.23 &0.25 &0.0177 \\
    \textbf{Crocodile} &11631 &360040 &5 &13183 &0.24 &0.30 &0.0036 \\
    \bottomrule
    \end{tabular}
    }
    \label{tb_Statistics}
\end{center}
\vspace{-1.5em}
\end{table}

\subsection{Neutral Pair in Contrastive Learning}\label{sec_neutral_pair}
Recognizing that graphs exhibit varying degrees of homophily, we identify two limitations in existing approaches: (1) Traditional CGC methods that ignore structural information often yield suboptimal performance (particularly on homophilic graphs), while (2) Neighbor-based CGC methods treating all neighbors as positive pairs suffer from heterophilic neighbor interference (especially on heterophilic graphs). To adapt to graphs with arbitrary homophily ratios, we pioneer the concept of neutral pair for contrastive learning, as described in Definition \ref{df_neutral_pair}.
\begin{definition}[Neutral Pair]\label{df_neutral_pair}
Given a node $v_i \in \mathcal{V}$ with its neighbor set $\mathcal{N}_i$ in graph $\mathcal{G}= (\mathcal{V}, \mathcal{E})$ and one of its neighbors $v_k \in \mathcal{N}_i$, the sample pair ($v_i$, $v_k$) forms a neutral pair in our contrastive learning framework. Distinct from conventional positive/negative pairs, each neutral pair is treated as a partial positive pair weighted by $w \in [0, 1]$, dynamically adjusting its contribution during contrastive learning.
\end{definition}

The neutral pair extends the concept of positive pair in traditional GCL without being fully equivalent. With respect to the weight $w$, simply treating it as a hyper-parameter lacks flexibility and is inefficient. A more effective strategy is to learn adaptive weights tailored to each graph. In this paper, we predefine a neutral contrastive factor, akins to a simulated homophily ratio, for coarse-grained weighting of all neutral pairs in Section \ref{sec_NCA}, and leverage edge-specific weights, for fine-grained weighting of individual neutral pairs in Section \ref{sec_AFC}. We integrate these two strategies into a unified framework to better exploit reliable neighborhood structural information, ultimately enhancing the model’s representation learning capability.

\subsection{Graph Neighborhood Congener Ratio}\label{sec_congener_ratio}
In a graph, we refer to nodes belonging to the same category as \textit{congeners}.
Based on this, we first propose a property for graph data, namely graph neighborhood congener ratio, denoted by $\delta$, which measures the average proportion of nodes' congeners are connected within their neighborhood in a graph.
\begin{definition}[Graph Neighborhood Congener Ratio]\label{df_delta}
Consider a node $v_i \in \mathcal{V}$ with its neighbor set $\mathcal{N}_i$, its neighborhood congener ratio is the fraction of congener number in neighbor set $\mathcal{N}_i$ to that in the whole node set $\mathcal{V}$ of $\mathcal{G}$. The graph neighborhood congener ratio $\delta$ is the average of all nodes' neighborhood congener ratio:
\begin{equation}\label{eq_delta}
    \delta  = \frac{1}{|\mathcal{V}|} \sum_{v_i \in \mathcal{V}} \frac{|\{v_k : v_k \in \mathcal{N}_i \wedge y_i = y_k\}|}{|\{v_j : v_j \in \mathcal{V} \wedge y_i = y_j\}|}.
\end{equation}
\end{definition}

The graph neighborhood congener ratio $\delta$ reflects the percentage that the nodes have the same semantic information are connected as neighbors. The larger value of $\delta$, the higher average proportion of neighboring congener nodes within the entire graph.
In practice, as shown in Table \ref{tb_Statistics}, the graph neighborhood congener ratios are remarkably low for most real-world graphs. Therefore, exploring additional semantically similar nodes for each node to construct neutral pairs is crucial for effective neighborhood contrastive learning.

\section{Methodology}
In this section, we elaborate on the proposed NeuCGC, including the pseudo-Siamese encoders (Section \ref{sec_encoders}), the global feature distribution alignment (Section \ref{sec_GDA}), the neutral contrastive distribution alignment (Section \ref{sec_NCA}) and the adaptive feature consistency neutral contrastive learning (Section \ref{sec_AFC}) modules. An overview of it is shown in Fig. \ref{fig_overview}.

\subsection{Pseudo-Siamese Encoders for Embedding Learning}
\label{sec_encoders}
Contrastive graph clustering generally learns more generalizable node representations by generating and contrasting different augmented views of the graph. However, some studies \cite{2022DSSL-xiao-decoupled, 2022AFGRL-Lee} have shown that improper augmentations can lead to semantic drift and degrade performance on heterophilic graphs. Additionally, graphs with low homophily can result in indistinguishable representations due to the dependency of graph neural networks (GNNs) on graph structure aggregation, which is more effective on homophilic graphs \cite{2022Hete-zheng},\cite{2025R4-li-hypergraph},\cite{2024R4-li-permutation}.

To address these challenges, and following the approach in \cite{2023CCGC-Yang}, we employ pseudo-Siamese encoders with dual parameter-unshared multilayer perceptrons (MLPs) for embedding learning. This technique helps mitigate these issues by extracting feature representations directly from node attributes to the latent space, which could be formally defined as:
\begin{equation}\label{eq_encoders}
\begin{split}
    &\textbf{Z}^{(l)} = \mathcal{F}_{\mathbf{\Theta}^{(l)}}({\textbf{X}}),
    \\s.t.~&\mathbf{\Theta}^{(1)} \cap \mathbf{\Theta}^{(2)} = \varnothing ,
\end{split}
\end{equation}
\noindent where $l\in \{1,2\}$, the same applies hereafter. $\textbf{Z}^{(1)}, \textbf{Z}^{(2)} \in \mathbb{R}^{N \times d}$ denote the learned latent graph embeddings from corresponding encoder, with different parameters \(\mathbf{\Theta}^{(1)}\) and \(\mathbf{\Theta}^{(2)}\) respectively, and $d$ is the dimension of the latent embeddings. 
As a result, $\mathcal{F}_{\mathbf{\Theta}^{(1)}}$ and $\mathcal{F}_{\mathbf{\Theta}^{(2)}}$ capture different semantic information from the original node attributes.
We then develop both global and local distribution alignment modules to distill consistent information from these two views and adaptive feature consistency neutral contrastive learning, thereby enhancing representation learning.

\subsection{Global Feature Distributions Alignment}
\label{sec_GDA}
The goal of contrastive learning is to align the latent semantic properties of two views, $\textbf{Z}^{(1)}$ and $\textbf{Z}^{(2)}$, which capture different semantic information. This process aims to reduce the discrepancy between these distributions, encouraging the model to learn more general and consistent information. 
To achieve this, we first introduce Global Feature Distribution Alignment (GDA) module, ensuring that global consistent information is extracted from the two semantically separated views. Specifically, we treat the representations of all nodes in each view as a global probability distribution, denoted by $p_{\mathbf{\Theta}^{(l)}}(\textbf{Z}^{(l)}|\textbf{X})$, 
Additionally, $\textbf{Z}^{(l)}=[\textbf{z}^{(l)}_1, \textbf{z}^{(l)}_2, ..., \textbf{z}^{(l)}_N]$ and $p_{\mathbf{\Theta}^{(l)}}(\textbf{z}^{(l)}_i|\textbf{x}_i)$ is treated as a separate node probability distribution of latent representation $\textbf{z}^{(l)}_i \in \mathbb{R}^{d}$ for node $v_i$.

Since the direction of divergence between two embedding distributions is not critical in our model, it is more reasonable to minimize the recently widely used Symmetric Kullback-Leibler (SKL) divergence \cite{2021SKL-RDrop},\cite{2021SKL-HGIB},\cite{2023SKL-GIST},\cite{2024SymKL-Yao}, $D_{SKL}(P,Q) = D_{KL}(P||Q) + D_{KL}(Q||P)$ where $P$ and $Q$ are two probability distributions and $D_{KL}(\cdot)$ is the KL divergence \cite{1951KL-kullback}, between them in a bidirectional manner. This approach aligns the feature distributions of the two views by minimizing the divergence between them. Specifically, we aim to align both the global representation space and the node-specific representation space across the two views by minimizing the SKL divergence between $p_{\mathbf{\Theta}^{(1)}}(\textbf{Z}^{(1)}|\textbf{X}) $ and $p_{\mathbf{\Theta}^{(2)}}(\textbf{Z}^{(2)}|\textbf{X}) $, as well as between $p_{\mathbf{\Theta}^{(1)}}(\textbf{z}^{(1)}_i|\textbf{x}_i)$ and $p_{\mathbf{\Theta}^{(2)}}(\textbf{z}^{(2)}_i|\textbf{x}_i)$. 
Subsequently, the global feature distribution alignment loss can be defined as:
\begin{equation}\label{eq_L_GDA}
\begin{split}
    \mathcal{L}_{GDA}= & D_{SKL}(p_{\mathbf{\Theta}^{(1)}}(\textbf{Z}^{(1)}|\textbf{X}), p_{\mathbf{\Theta}^{(2)}}(\textbf{Z}^{(2)}|\textbf{X}) ) \\ 
    &+ \sum_{v_i \in \mathcal{V}} D_{SKL}(p_{\mathbf{\Theta}^{(1)}}(\textbf{z}^{(1)}_i|\textbf{x}_i), p_{\mathbf{\Theta}^{(2)}}(\textbf{z}^{(2)}_i|\textbf{x}_i)).
\end{split}
\end{equation}

By minimizing $\mathcal{L}_{GDA}$, the model is encouraged to achieve global mutual distribution alignment between the two views of the latent embeddings. This technique helps alleviate the inconsistency between $\textbf{Z}^{(1)}$ and $\textbf{Z}^{(2)}$, ensuring information sharing from a global perspective.

\subsection{Neutral Contrastive Distribution Alignment}
\label{sec_NCA}
A key limitation of existing neighbor-based CGC methods is their reliance on the homophily assumption, where each node and its neighbors are treated as positive pairs. This becomes problematic in low homophily graphs, where connected nodes may belong to different classes.
To address this, we first introduce neutral pair among neighbor nodes in Definition \ref{df_neutral_pair}, and propose a neutral contrastive learning mechanism. The neutral pairs are adaptively selected as positive pairs based on the graph’s homophily level. In low homophily graphs, neutral pairs contribute less to contrastive learning, while in high homophily graphs, they contribute more. We treat neighbors as neutral pairs, distinguishing them from positive pairs, and apply a \textit{neutral contrastive factor (NCF)} ranging from 0 to 1 to weight neutral pairs as positive. This mechanism completely enhances consistency for positive pairs, partially strengthens neutral pairs, and reduces that for negative pairs, encouraging the model to learn discriminative features.

Building on this motivation, in addition to achieving global alignment between the semantically separated feature distributions, we also focus on node neighborhoods from a local perspective, informed by the graph's homophily level. To achieve this, we define the neighborhood distribution neutral contrastive alignment (NCA) loss as follows:
\begin{equation}\label{eq_L_NCA}
    \mathcal{L}_{NCA}=\frac{1}{|\mathcal{V}|} \sum_{v_i \in \mathcal{V}}\frac{\frac{1}{|\mathcal{N}_i|+1} (\textbf{K}_{ii} + \eta \sum_{v_k \in {\mathcal{N}_i}} \textbf{K}_{ik})}
{\frac{1}{|\mathcal{V}|-1}\sum_{v_j\in \mathcal{V} \wedge  j \ne i}\textbf{K}_{ij}},
\end{equation}
where $\eta$ is the neutral contrastive factor and $\textbf{K} \in \mathbb{R}^{N \times N}$ is the node-level pairwise cross-view SKL divergence matrix, defined as:
\begin{equation}\label{eq_cmp_K}
    \textbf{K}_{ij}=D_{SKL}(p_{\mathbf{\Theta}^{(1)}}(\textbf{z}^{(1)}_i|\textbf{x}_i), p_{\mathbf{\Theta}^{(2)}}(\textbf{z}^{(2)}_j|\textbf{x}_j)).
\end{equation}
Due to the non-negativity of KL and SKL divergences, for nodes $v_i, v_j \in \mathcal{V}$, we have $\textbf{K}_{ij} > 0$.
The numerator of $\mathcal{L}_{NCA}$ represents the average SKL divergence between positive pairs and weighted neutral pairs, while the denominator captures the divergence between non-positive pairs.

The key challenge of the proposed neighborhood distribution neutral contrastive alignment lies in estimating the neutral contrastive factor $\eta$.
To address this, we propose a \textbf{Neutral Contrastive Factor Estimation (NCFE)} strategy. Specifically, we define $\textbf{S}_{ij}$ as cross-view similarity between nodes $v_i$ and $v_j$, based on embeddings $\textbf{Z}^{(1)}$ and $\textbf{Z}^{(2)}$: 
\begin{equation}\label{eq_cosine_sim}
	\textbf{S}_{ij}=\frac{\mathbf{Z}^{(1)}_i (\mathbf{Z}^{(2)}_j)^{\top}}{||\mathbf{Z}^{(1)}_i||_2||\mathbf{Z}^{(2)}_j||_2}.
\end{equation}
We compute the min-max scaled similarity $\textbf{S}^{\mathcal{N}} \in \mathbb{R}^{N \times N}$ for the original neighbors in $\textbf{A}$ as:
\begin{equation}\label{eq_cmp_SN} \textbf{S}^{\mathcal{N}}=\mathrm{norm}(\textbf{S}) \odot \textbf{A}, \end{equation}
\noindent where $\mathrm{norm}(\cdot)$ is the min-max normalization, and $\odot$ denotes the Hadamard product. Here, $\textbf{S}^{\mathcal{N}}_{ik} \in [0, 1]$ quantifies the semantic similarity between node $v_i$ and its neighbor $v_k \in \mathcal{N}_i$.
Higher similarity typically implies stronger semantic consistency. Based on this, we hypothesize that within a node's neighborhood, neighbors with higher similarity are more likely to belong to the same category. Specifically, if $\textbf{S}^{\mathcal{N}}_{ik} \ge \xi$ (where $\xi$ is a predefined threshold), $v_k$ is assumed to share the same category as $v_i$.
The neutral contrastive factor $\eta$ is then calculated as the average fraction of same-category neighbors within each node's neighborhood in a \textit{coarse-grained} manner:
\begin{equation}\label{eq_cmp_eta}
	\eta =\frac{1}{|\mathcal{V}|} \sum_{v_i \in \mathcal{V}} \frac{\sum_{k\ne i} \mathds{1}(v_k \in \mathcal{N}_i\wedge \textbf{S}^{\mathcal{N}}_{ik} \ge \xi )}{|\mathcal{N}_i|}, 
\end{equation}
\noindent where $\mathds{1}(\cdot)$ represents the indicator function. The threshold $\xi$ determines the similarity level at which two neighbors are deemed to share semantic label, adaptively defined as:
\begin{equation}
\xi = \frac{\mathrm{tr}(\mathrm{norm}(\textbf{S}))}{|\mathcal{V}|},
\end{equation}
\noindent where $\mathrm{tr}(\cdot)$ denotes the trace operator, which computes the sum of main diagonal elements of a square matrix.

The objective of the NCA loss is to align the probability distributions of positive pairs and weighted neutral pairs by minimizing the numerator while pushing apart the distributions of non-positive pairs by maximizing the denominator. By minimizing $\mathcal{L}_{NCA}$, the model achieves neighborhood feature distribution alignment and facilitates information interaction from a local perspective, promoting the learning of more general semantic information.


\subsection{Adaptive Feature Consistency Neutral Contrastive Learning}
\label{sec_AFC}
To enhance structural consistency in $\textbf{Z}^{(1)}$ and $\textbf{Z}^{(2)}$ at the feature level, we propose Adaptive Feature Consistency Neutral Contrastive Learning (AFC) by considering the graph neighborhood congener ratio. As highlighted in Section \ref{sec_congener_ratio}, real-world graphs often have low neighborhood congener ratios, meaning a large number of congeners exist outside a node's immediate neighborhood. Including these distant congeners as neighbors can mitigate the negative impact of low homophily and significantly improve the model's ability to learn consistent features.
For example, in the original graph $\textbf{A}$ illustrated in Fig. \ref{fig_overview}, node $v_1$ possesses additional, albeit not immediately visible, semantically similar distant neighbors such as $v_7$, $v_8$, $v_9$, and others, which can be identified and incorporated as new neighbors. These connections facilitate more effective learning of consistent information.

To achieve this, we expand the neutral pairs to enhance the model's discriminative power. Specifically, clustering results $\textbf{c} \in \mathbb{R}^{N}$ are generated by applying the K-means algorithm to the concatenated graph embeddings:
\begin{equation}\label{eq_fusion}
	\textbf{Z} = [\textbf{Z}^{(1)}, \textbf{Z}^{(2)}].
\end{equation}
\noindent We then select the top $\textit{k}$ high-confidence pseudo labels $\textbf{c}^h \in \mathbb{R}^{k \cdot N}$, where $\textit{k} \in (0, 1)$ is a confidence signal to reconstruct the neighborhood for each node. The confidence is measured by the distance to the cluster center. Additionally, we assign a learned weight to each of the original neighbors. Given a node pair $v_i, v_j$ where $i \ne j$, the high-confidence graph $\textbf{H} \in \mathbb{R}^{N \times N}$ can be defined as:
\begin{equation}\label{eq_cmp_H}
    \textbf{H}_{ij} = \begin{cases}
    1, & \textbf{c}^h_i = \textbf{c}^h_j, \\
    \mathrm{norm}(\textbf{S}_{ij}),~ & \textbf{A}_{ij}=1 \wedge \textbf{c}^h_i \ne \textbf{c}^h_j, \\
    0, & otherwise,
\end{cases}
\end{equation}
\noindent where $\textbf{S}_{ij}$ is the cosine similarity between paired latent node embeddings. $\textbf{H}_{ij}$ ranges from 0 to 1. Specifically, if $0 < \textbf{H}_{ij} \le 1$, then $v_i$ and $v_j$ are either high-confidence neighbors or original neighbors, and are treated as a weighted neutral pair. Otherwise, they are considered negative pair. 
Using the high-confidence graph $\textbf{H}$, we identify neutral pairs and define \textit{fine-grained} neutral contrastive factors via edge weights. Thus, by modifying the widely used InfoNCE loss \cite{2018CPC-oord}, the adaptive feature consistency neutral contrastive loss is given by: 
\begin{equation}\label{eq_L_AFC}
\begin{split}
    \mathcal{L}_{AFC}= \frac{1}{|\mathcal{V}|} \sum_{v_i \in \mathcal{V}} 
    -\mathrm{log}
    \frac{\mathrm{exp}(\textbf{S}_{ii}) 
    + \sum_{v_k \in \widehat{\mathcal{N}}_i} \textbf{H}_{ik} \mathrm{exp}(\textbf{S}_{ik}) } {\sum_{v_j \in \mathcal{V}} \mathrm{exp}(\textbf{S}_{ij})},
\end{split}
\end{equation}
\noindent where $\widehat{\mathcal{N}}_i$ is the high-confidence neighbor set of $v_i$ in $\textbf{H}$. 

By minimizing $\mathcal{L}_{AFC}$, we maximizing the similarity of positive pairs, partially maximizing those of neutral pairs, while minimizing those of negative pairs.
Therefore, by leveraging reliable self-supervised signals, AFC enhances the model's discriminative power and consistency in feature learning, resulting in higher-quality graph embeddings.
\begin{algorithm}[!t]
\caption{Optimization of NeuCGC}
\label{algorithm}
\textbf{Input}: the input graph $\mathcal{G} = \{\textbf{X}, \textbf{A}\}$; cluster number $C$; hyper-parameters $\textit{k}, \lambda_1, \lambda_2$; iteration number $I$.

\begin{algorithmic}[1] 
\WHILE{epoch = 1, 2, ... $I$}
\STATE Obtain the latent embeddings $\textbf{Z}^{(1)}, \textbf{Z}^{(2)}$ by encoding $\mathcal{G}$ with \textcolor{blue}{(\ref{eq_encoders})}.
\STATE Calculate the GDA loss $\mathcal{L}_{GDA}$ with \textcolor{blue}{(\ref{eq_L_GDA})}.
\STATE Calculate $\eta$ with \textcolor{blue}{(\ref{eq_cmp_eta})}, the pairwise SKL divergence matrix $\textbf{K}$ with \textcolor{blue}{(\ref{eq_cmp_K})}, and the NCA loss $\mathcal{L}_{NCA}$ based on the original graph $\textbf{A}$ with \textcolor{blue}{(\ref{eq_L_NCA})}.
\STATE Obtain high-confidence pseudo labels $\textbf{c}^h$ from cluster assignment $\textbf{c}$ by performing K-means algorithm on the fused embeddings $\textbf{Z}$ with \textcolor{blue}{(\ref{eq_fusion})}, then calculate the high-confidence graph $\textbf{H}$ by \textcolor{blue}{(\ref{eq_cmp_H})}.
\STATE Calculate the AFC loss $\mathcal{L}_{AFC}$ based on $\textbf{H}$ with \textcolor{blue}{(\ref{eq_L_AFC})}.
\STATE Update the entire networks by minimizing the final objective $\mathcal{L}$.
\ENDWHILE
\end{algorithmic}
\textbf{Output}: Clustering results $\textbf{c}$.
\end{algorithm}

\subsection{Optimization}
\label{sec_optimization}
In summary, we jointly optimize the three objectives above, and the final objective function can be written as:
\begin{equation}\label{eq_L_total}
\begin{split}
    \mathop{\min}\limits_{\mathbf{\Theta}^{(1)}, \mathbf{\Theta}^{(2)}} \mathcal{L} = \mathop{\min}\limits_{\mathbf{\Theta}^{(1)}, \mathbf{\Theta}^{(2)}} \mathcal{L}_{NCA} + \lambda_1 \mathcal{L}_{AFC} + \lambda_2 \mathcal{L}_{GDA},
\end{split}
\end{equation}
\noindent where $\lambda_1$ and $\lambda_2$ are factors for balancing the losses. 
Fig. \ref{fig_overview} provides an overview of the proposed framework NeuCGC, and its optimization algorithm can be found in Algorithm \ref{algorithm}. 

\noindent \textbf{Computational Complexity Analysis.}
The computational complexity of our approach could be discussed from two perspectives: the pseudo-Siamese networks and the three proposed objective functions. For the former, the time complexity is $\mathcal{O}(NDd)$, and the space complexity is $\mathcal{O}(Dd)$, where $D$ and $d$ are the dimensions of node attributes and latent embeddings respectively, and $N$ is the sample size. 
In terms of the objective functions, the complexity mainly depends on the computations of $\mathcal{L}_{GDA}$, $\mathcal{L}_{NCA}$ and $\mathcal{L}_{AFC}$.
Computing $\mathcal{L}_{GDA}$ requires both time and space complexities of $\mathcal{O}(Nd)$. The time and space complexities for computing $\textbf{S}$, $\textbf{S}^{\mathcal{N}}$, $\textbf{K}$, and $\textbf{H}$ are $\mathcal{O}(N^2d)$ and $\mathcal{O}(N^2)$, respectively. 
These computations result in the time complexity of $\mathcal{O}(N^2d)$ and space complexity of $\mathcal{O}(N^2)$ for both $\mathcal{L}_{NCA}$ and $\mathcal{L}_{AFC}$. 

Consequently, the overall time and space complexities of NeuCGC are $\mathcal{O}(N^2d)$ and $\mathcal{O}(N^2)$, respectively. This is comparable to that of the classical InfoNCE loss.

\subsection{Discussion}
In this subsection, we provide analyses of why the proposed global and local feature distribution alignment schemes facilitate the model's representation learning, and why the introduced AFC module enhances feature consistency, respectively.

\textit{1) Why do the proposed global and local feature distribution alignment schemes facilitate representation learning?} 
As aforementioned, the output $\textbf{Z}^{(1)}$ and $\textbf{Z}^{(2)}$ of the pseudo-Siamese encoders are two semantically distinct feature distributions. They both capture common information and specific knowledge from the original graph attributes. Therefore, we perform distribution alignment between them with GDA and NCA objectives from global and local perspectives respectively, so as to achieve knowledge sharing. 
The KL loss \cite{1951KL-kullback, 2015KD-Hinton} is frequently selected as an effective technique for achieving knowledge distillation due to its minimization of the KL divergence between two distinct probability distributions.
Meanwhile, minimizing the SKL divergence \cite{2021SKL-RDrop},\cite{2021SKL-HGIB},\cite{2023SKL-GIST},\cite{2024SymKL-Yao} is equivalent to simultaneously minimizing the bidirectional KL divergence for two probability distributions $P$ and $Q$, formulated as:
\begin{equation}\label{eq_minSKL}
    \mathop{\min}\limits_{\mathbf{\Theta}^{(1)}, \mathbf{\Theta}^{(2)}} D_{SKL}(P, Q) = \mathop{\min}\limits_{\mathbf{\Theta}^{(1)}, \mathbf{\Theta}^{(2)}} D_{KL}(P || Q) + D_{KL}(Q || P).
\end{equation}
\noindent We minimize the SKL divergence between $\textbf{Z}^{(1)}$ and $\textbf{Z}^{(2)}$ to enable knowledge transfer with each other and ultimately achieve semantic information sharing. 

For the GDA objective, minimizing it is equivalent to minimizing the bidirectional KL divergence between different view distributions, 
as well as between different instance distributions of the same node.
By this way, the global consistency is strengthened.
Regarding the NCA objective, we consider different instance distributions of the same node $(v^{(1)}_i, v^{(2)}_i)$ as positive pair, each node and its original neighbors $(v^{(1)}_i, v^{(2)}_k)$ as neutral pairs, while each node and its non-neighbors $(v^{(1)}_i, v^{(2)}_j)$ as negative pairs. Since $\mathcal{L}_{NCA}$ is proportional to the positive pairs and the weighted neutral pairs, minimizing this objective is equivalent to minimizing the bidirectional KL divergence between positive pair and partially between neutral pairs, while maximizing that between negative pairs. Thus, we promote information interaction within each node's neighborhood and the local neighborhood consistency is improved.
Through global and local distribution alignment, we encourage the model to capture more comprehensive semantic information and achieve more consistent representations.

\textit{2) Why does the introduced AFC module enhance feature consistency?}
The AFC module expand neutral pairs via discovering high-quality neighbors, and each neutral pair is weighted with a fine-grained neutral contrastive factor. This self-supervised signal ensures that the model learns richer and more generalized consistency features.
Here, we analyze the optimization of $\mathcal{L}_{AFC}$ from the perspective of mutual information maximization, which has been widely employed in previous studies \cite{2018CPC-oord},\cite{2019DGI-velickovic},\cite{2019I-NCE-poole19a},\cite{2019MImax-Bachman},\cite{2020GRACE-Zhu},\cite{2020MuImax-Tschannen},\cite{2021GCA-R2-zhu}.
Let two random variables $\textbf{Z}^{(1)}, \textbf{Z}^{(2)}$ denote as the latent node representations output by the two pseudo-Siamese encoders respectively, with $\textbf{X}$ as their input graph attributes. 
We have the following theorem:
\begin{theorem}[]\label{thm:MI}
Minimizing $\mathcal{L}_{AFC}$ in \textcolor{blue}{(\ref{eq_L_AFC})} incorporates the maximization of the InfoNCE objective $I_{NCE}$, which is equivalent to maximizing mutual information between the original attributes $\mathbf{X}$ and the two latent representations $\mathbf{Z}^{(1)}$ and $\mathbf{Z}^{(2)}$:
\begin{equation}\label{eq_theorem}
\begin{split}
    -\mathcal{L}_{AFC} \le I_{NCE} \le I(\mathbf{X};\mathbf{Z}^{(1)}, \mathbf{Z}^{(2)}).
\end{split}
\end{equation}
Therefore, optimizing the AFC objective can lead to superior representations compared to the InfoNCE objective.
\end{theorem}
\begin{proof}\renewcommand{\qedsymbol}{} 
Define the critic function as $\theta(\textbf{z}^{(1)}_i, \textbf{z}^{(2)}_j) = \theta_{ij} = \textbf{S}_{ij}$. Then, the negative AFC objective $-\mathcal{L}_{AFC}$ can be rewritten as 
\begin{equation}\label{L_AFC_rewritten}
\begin{split}
    -\mathcal{L}_{AFC} &= \mathds{E} \left [ \frac{1}{|\mathcal{V}|} \sum_{v_i \in \mathcal{V}} 
    \mathrm{log} \frac{\mathrm{exp}(\theta_{ii}) 
    + \sum_{v_k \in \widehat{\mathcal{N}}_i} \textbf{H}_{ik} \mathrm{exp}(\theta_{ik}) } {\sum_{v_j \in \mathcal{V}} \mathrm{exp}(\theta_{ij})} \right ] \\ 
    &\le \mathds{E} \left [ \frac{1}{|\mathcal{V}|} \sum_{v_i \in \mathcal{V}} 
    \mathrm{log} \frac{\mathrm{exp}(\theta_{ii}) 
    + \sum_{v_k \in \widehat{\mathcal{N}}_i} \mathrm{exp}(\theta_{ik}) } {\sum_{v_j \in \mathcal{V}} \mathrm{exp}(\theta_{ij})} \right ] \\
    &\le \mathds{E} \left [ \frac{1}{|\mathcal{V}|} \sum_{v_i \in \mathcal{V}} 
    \mathrm{log} \frac{|\mathcal{V}| \mathrm{exp}(\theta_{ii}) 
    } {\sum_{v_j \in \mathcal{V}} \mathrm{exp}(\theta_{ij})} \right ].
\end{split}
\end{equation}
\noindent These two inequalities hold because we take into account the neutral pairs with $0 \le \textbf{H}_{ij} \le 1$ and most of the real-world graphs are sufficiently sparse, respectively. In addition, according to \cite{2018CPC-oord, 2019I-NCE-poole19a, 2020MuImax-Tschannen}, the InfoNCE objective can be defined as
\begin{equation}\label{eq_INCE}
\begin{split}
    I_{NCE} \triangleq \mathds{E} \left [ \frac{1}{N} \sum_{i=1}^N 
    \mathrm{log} \frac{\mathrm{exp}(\theta(\textbf{z}^{(1)}_i, \textbf{z}^{(2)}_i))} {\frac{1}{N} \sum_{j=1}^N \mathrm{exp}(\theta(\textbf{z}^{(1)}_i, \textbf{z}^{(2)}_j))} \right ],
\end{split}
\end{equation}
which is the lower bound of the mutual information of representations $\textbf{Z}^{(1)}, \textbf{Z}^{(2)}$, namely,
\begin{equation}\label{eq_INCE_IZ}
    I_{NCE} \le I(\textbf{Z}^{(1)}, \textbf{Z}^{(2)}).
\end{equation}
Note that $|\mathcal{V}|=N$, and this leads to the following inequalities:
\begin{equation}\label{eq_AFC_IZ}
    -\mathcal{L}_{AFC} \le I_{NCE} \le I(\textbf{Z}^{(1)}, \textbf{Z}^{(2)}).
\end{equation}
As demonstrated in literature \cite{2020GRACE-Zhu}, we have the following inequality with respect to mutual information between the original attributes $\textbf{X}$ and the two latent representation views $\textbf{Z}^{(1)}$ and $\textbf{Z}^{(2)}$,
\begin{equation}\label{eq_IZ_IXZ}
    I(\textbf{Z}^{(1)}, \textbf{Z}^{(2)}) \le I(\textbf{X}; \textbf{Z}^{(1)}, \textbf{Z}^{(2)}).
\end{equation}
At last, we arrive at the required inequalities
\begin{equation}\label{eq_AFC_INCE_IXZ}
    -\mathcal{L}_{AFC} \le I_{NCE} \le I(\textbf{X}; \textbf{Z}^{(1)}, \textbf{Z}^{(2)}).
\end{equation}
Evidenced by literature \cite{2020MuImax-Tschannen}, maximizing looser lower bounds on mutual information can lead to better representations. Inequalities \textcolor{blue}{(\ref{eq_AFC_INCE_IXZ})} prove that $-\mathcal{L}_{AFC}$ serves as a lower bound of $I_{NCE}$, which is a stricter estimator of the mutual information between the original attributes and the two latent node representations. Consequently, minimizing $\mathcal{L}_{AFC}$ involves the maximization of $I_{NCE}$, simultaneously can be seen as a proxy for maximizing $I(\textbf{X}; \textbf{Z}^{(1)}, \textbf{Z}^{(2)})$, and finally results in superior representations. 
Thus, the proof is completed.
\end{proof}
\begin{table*}[!ht]
    \renewcommand{\arraystretch}{1.15}
    \caption{Comparison of average clustering performance under 10 runs on eight graph datasets (as well as their homophily ratios $r_h$) obtained by our approach to sixteen state-of-the-art methods. Four metrics (in \%) with mean and std are used to evaluate the clustering results. $\mathbf{Average}$ represents the mean score of each method across eight datasets. The best and second best performance on each dataset are marked with boldface and underline, respectively.}
    \begin{center}
    \setlength{\tabcolsep}{0.7mm}{
    \begin{tabular}{c|cccc|cccc|cccc}
    \hline \hline 
    \multirow{2}*{\textbf{Method}} & \multicolumn{4}{c|}{\textbf{Cora~(0.81)}} & \multicolumn{4}{c|}{\textbf{Citeseer~ (0.74)}} & \multicolumn{4}{c}{\textbf{ACM~(0.82)}} \\
     &\textbf{ACC} &\textbf{NMI} &\textbf{ARI} &\textbf{F1} &\textbf{ACC} &\textbf{NMI} &\textbf{ARI} &\textbf{F1} &\textbf{ACC} &\textbf{NMI} &\textbf{ARI} &\textbf{F1} \\
    \hline
    \textbf{SDCN}~(\textit{WWW}'20)\cite{2020SDCN-bo}
    &56.0$\pm$2.6 &34.6$\pm$3.0 &29.2$\pm$2.3 &48.6$\pm$4.1 
    &65.9$\pm$0.3 &38.7$\pm$0.3 &40.1$\pm$0.4 &63.6$\pm$0.2 
    &90.0$\pm$0.5 &66.8$\pm$1.5 &72.7$\pm$1.3 &90.0$\pm$0.5 \\
    \textbf{DFCN}~(\textit{AAAI}'21)\cite{2021DFCN-tu-deep} 
    &72.6$\pm$0.6 &53.6$\pm$1.0 &49.3$\pm$1.1 &68.0$\pm$0.7 
    &69.9$\pm$0.2 &44.1$\pm$0.3 &46.0$\pm$0.3 &64.7$\pm$0.2 
    &90.9$\pm$0.1 &69.7$\pm$0.3 &75.1$\pm$0.3 &90.8$\pm$0.1 \\
    \textbf{AUTOSSL}~(\textit{ICLR}'22)\cite{2022AutoSSL-jin} 
    &63.5$\pm$0.7 &47.1$\pm$2.1 &39.4$\pm$1.6 &57.9$\pm$3.4 
    &68.4$\pm$0.2 &42.7$\pm$0.1 &43.9$\pm$0.2 &63.6$\pm$0.2 
    &90.3$\pm$0.4 &68.1$\pm$0.7 &73.5$\pm$1.0 &90.3$\pm$0.3 \\
    \textbf{GraphMAE}~(\textit{KDD}'22)\cite{2022GraphMAE-hou} 
    &74.0$\pm$0.5 &54.7$\pm$0.9 &50.2$\pm$1.1 &73.3$\pm$0.9
    &70.3$\pm$0.3 &44.0$\pm$0.5 &46.0$\pm$0.6 &\underline{65.9$\pm$0.3}
    &91.1$\pm$0.3 &69.3$\pm$0.6 &75.5$\pm$0.8 &91.1$\pm$0.3 \\
    \textbf{DyFSS}~(\textit{AAAI}'24)\cite{2024DyFSS-zhu} 
    &72.1$\pm$0.2 &53.6$\pm$0.3 &49.6$\pm$0.4 &70.5$\pm$0.2 
    &68.3$\pm$0.4 &41.6$\pm$0.2 &42.9$\pm$0.5 &64.0$\pm$0.5 
    &87.8$\pm$0.1 &61.0$\pm$0.3 &67.3$\pm$0.4 &87.9$\pm$0.1 \\
    \textbf{DGCluster}~(\textit{AAAI}'24)\cite{2024DGCluster-bhowmick} 
    &76.6$\pm$3.9 &\textbf{64.2$\pm$2.4} &\textbf{61.8$\pm$3.4} &67.9$\pm$7.3 
    &67.6$\pm$4.5 &45.0$\pm$3.0 &40.5$\pm$5.7 &64.6$\pm$5.0 
    &88.8$\pm$3.8 &68.9$\pm$8.0 &69.9$\pm$10.1 &88.9$\pm$3.7 \\
    \hline
    \textbf{DCRN}~(\textit{AAAI}'22)\cite{2022DCRN-liu} 
    &73.8$\pm$0.1 &56.3$\pm$0.2 &53.6$\pm$0.2 &64.9$\pm$0.2 
    &\underline{71.1$\pm$0.1} &\underline{46.0$\pm$0.1}
    &\underline{47.8$\pm$0.1} &\textbf{66.2$\pm$0.1} 
    &92.0$\pm$0.1 &72.8$\pm$0.3 &78.6$\pm$0.2 &\underline{92.3$\pm$0.1} \\
    \textbf{AGC-DRR}~(\textit{IJCAI}'22)\cite{2022AGC-DRR-Gong} 
     &62.4$\pm$0.6 &47.8$\pm$1.4 &45.0$\pm$1.9 &40.7$\pm$0.2 
     &68.1$\pm$1.6 &43.5$\pm$1.3 &45.3$\pm$1.9 &64.7$\pm$1.0 
     &\underline{92.4$\pm$0.1} &\underline{72.1$\pm$0.3} &\underline{78.8$\pm$0.6} &91.8$\pm$0.1 \\
     \textbf{CONVERT}~(\textit{MM}'23)\cite{2023CONVERT-yang} 
     &73.9$\pm$1.0 &55.6$\pm$1.0 &50.2$\pm$1.5 &72.1$\pm$2.8 
     &68.2$\pm$0.7 &41.1$\pm$0.9 &42.6$\pm$1.3 &62.4$\pm$2.0 
     &86.2$\pm$0.8 &56.6$\pm$1.6 &63.4$\pm$1.9 &86.3$\pm$0.8 \\
    \textbf{SCGC}~(\textit{TNNLS}'23)\cite{2023SCGC-Liu} 
     &73.2$\pm$1.5 &55.6$\pm$0.9 &51.5$\pm$1.9 &70.3$\pm$2.1 
     &70.9$\pm$0.7 &45.0$\pm$0.5 &46.1$\pm$1.0 &62.1$\pm$0.9 
     &89.8$\pm$0.5 &66.7$\pm$1.1 &72.4$\pm$1.3 &89.7$\pm$0.5 \\
     \textbf{NCLA}~(\textit{AAAI}'23)\cite{2023NCLA-shen-neighbor} 
     &73.3$\pm$1.8 &57.9$\pm$1.2 &53.6$\pm$2.5 &71.8$\pm$1.3 
     &67.0$\pm$1.8 &42.4$\pm$0.8 &42.7$\pm$2.3 &63.0$\pm$0.7 
     &91.7$\pm$0.2 &71.5$\pm$0.6 &77.1$\pm$0.6 &91.7$\pm$0.2 \\
    \textbf{SCAGC}~(\textit{TMM}'23)\cite{2023SCAGC-Xia-TMM} 
     &64.7$\pm$0.2 &49.3$\pm$0.1 &40.5$\pm$0.3 &63.2$\pm$0.2 
     &63.4$\pm$0.1 &39.9$\pm$0.1 &39.4$\pm$0.1 &61.1$\pm$0.1 
     &91.8$\pm$0.1 &71.1$\pm$0.1 &77.3$\pm$0.1 &91.9$\pm$0.1 \\
    \textbf{CCGC}~(\textit{AAAI}'23)\cite{2023CCGC-Yang} 
     &74.4$\pm$1.2 &57.0$\pm$1.1 &52.4$\pm$1.7 &70.8$\pm$3.5 
     &69.6$\pm$1.4 &43.7$\pm$1.8 &44.6$\pm$2.4 &62.2$\pm$2.4 
     &90.1$\pm$0.4 &67.2$\pm$1.3 &73.1$\pm$1.2 &90.1$\pm$0.4 \\
    \textbf{HSAN}~(\textit{AAAI}'23)\cite{2023HSAN-liu} 
     &\underline{77.0$\pm$0.6} &\underline{59.1$\pm$0.8} &\underline{57.0$\pm$1.2} &\underline{75.3$\pm$1.2} &71.0$\pm$0.7 &45.2$\pm$0.9 &47.3$\pm$1.4 &63.2$\pm$2.1 
     &89.7$\pm$0.2 &66.8$\pm$0.6 &72.3$\pm$0.6 &89.7$\pm$0.2 \\
    \textbf{GraphACL}~(\textit{NIPS}'23)\cite{2023GraphACL-Xiao} 
     &72.9$\pm$0.8 &56.3$\pm$0.6 &50.3$\pm$1.3 &69.8$\pm$1.8 
     &69.9$\pm$0.5 &44.3$\pm$0.6 &46.2$\pm$1.0 &64.5$\pm$1.2 
     &88.3$\pm$0.5 &63.2$\pm$1.2 &68.7$\pm$1.2 &88.3$\pm$0.4 \\ 
     \textbf{HeterGCL}~(\textit{IJCAI}'24)\cite{2024HeterGCL-Wang} 
     & 65.1$\pm$3.8 & 48.6$\pm$2.6 & 42.6$\pm$4.9 & 57.7$\pm$4.0 
     & 68.4$\pm$1.4 & 43.2$\pm$1.4 & 42.7$\pm$2.3 & 59.3$\pm$1.2 
     & 87.6$\pm$2.0 & 63.4$\pm$3.1 & 66.9$\pm$4.8 & 87.6$\pm$2.0 \\
     \hline
    \textbf{NeuCGC}
     &\textbf{77.1$\pm$1.0} &59.0$\pm$1.1 &56.3$\pm$1.4 &\textbf{75.8$\pm$1.2} 
     &\textbf{72.5$\pm$0.5} &\textbf{46.7$\pm$0.7} &\textbf{48.1$\pm$0.9} &64.0$\pm$0.9 
     &\textbf{92.6$\pm$0.2} &\textbf{73.8$\pm$0.5} &\textbf{79.3$\pm$0.5} &\textbf{92.6$\pm$0.2} \\
    \midrule[1.0pt]
    \multirow{2}*{\textbf{Method}} 
    &\multicolumn{4}{c|}{\textbf{DBLP~(0.80)}} 
    &\multicolumn{4}{c|}{\textbf{Wisconsin~(0.18)}} 
    &\multicolumn{4}{c}{\textbf{Cornell~(0.30)}} \\
    &\textbf{ACC} &\textbf{NMI} &\textbf{ARI} &\textbf{F1} &\textbf{ACC} &\textbf{NMI} &\textbf{ARI} &\textbf{F1} &\textbf{ACC} &\textbf{NMI} &\textbf{ARI} &\textbf{F1} \\
    \hline
    \textbf{SDCN}~(\textit{WWW}'20)\cite{2020SDCN-bo}
    &68.0$\pm$1.8 &39.5$\pm$1.3 &39.1$\pm$2.0 &67.7$\pm$1.5 
    &42.6$\pm$0.5 &11.9$\pm$0.1 &12.7$\pm$0.5 &30.7$\pm$0.1 
    &52.3$\pm$2.1 &16.2$\pm$0.9 &12.3$\pm$1.0 &31.9$\pm$2.3 \\
    \textbf{DFCN}~(\textit{AAAI}'21)\cite{2021DFCN-tu-deep} 
    &76.5$\pm$0.4 &44.5$\pm$0.6 &47.6$\pm$1.1 &76.2$\pm$0.4 
    &51.7$\pm$1.7 &15.0$\pm$2.0 &17.6$\pm$2.7 &39.3$\pm$1.7 
    &56.3$\pm$1.4 &10.4$\pm$2.7 &14.2$\pm$3.6 &32.7$\pm$1.6 \\
    \textbf{AUTOSSL}~(\textit{ICLR}'22)\cite{2022AutoSSL-jin} 
    &48.7$\pm$3.1 &25.1$\pm$2.0 &13.5$\pm$2.7 &49.6$\pm$3.2 
    &49.8$\pm$2.2 &17.7$\pm$3.1 &13.9$\pm$3.3 &32.2$\pm$1.9 
    &46.7$\pm$1.6 &13.4$\pm$5.7 &11.0$\pm$2.5 &29.3$\pm$2.3 \\
    \textbf{GraphMAE}~(\textit{KDD}'22)\cite{2022GraphMAE-hou} 
    &\underline{80.4$\pm$0.5} &\underline{50.2$\pm$1.0} &\underline{55.4$\pm$1.0} &\underline{79.9$\pm$0.5} 
    &54.6$\pm$2.3 &19.0$\pm$4.0 &21.2$\pm$3.5 &40.6$\pm$3.5 
    &55.1$\pm$0.4 &4.3$\pm$0.8 &0.3$\pm$0.9 &17.4$\pm$0.4 \\
    \textbf{DyFSS}~(\textit{AAAI}'24)\cite{2024DyFSS-zhu} 
    &60.6$\pm$2.8 &29.7$\pm$2.1 &23.2$\pm$2.3 &61.0$\pm$2.6 
    &48.6$\pm$1.1 &11.3$\pm$1.1 &10.8$\pm$1.2 &32.2$\pm$1.3 
    &46.3$\pm$1.1 &8.0$\pm$0.4 &5.5$\pm$0.7 &24.6$\pm$1.0 \\
    \textbf{DGCluster}~(\textit{AAAI}'24)\cite{2024DGCluster-bhowmick} 
    &74.1$\pm$5.8 &45.9$\pm$4.0 &48.1$\pm$6.3 &70.7$\pm$9.1 
    &42.5$\pm$3.5 &12.3$\pm$2.5 &8.3$\pm$2.9 &32.2$\pm$3.3 
    &44.9$\pm$3.7 &11.0$\pm$4.2 &13.1$\pm$6.6 &31.0$\pm$2.1 \\
    
    \hline
    \textbf{DCRN}~(\textit{AAAI}'22)\cite{2022DCRN-liu} 
    &79.7$\pm$0.1 &49.2$\pm$0.5 &53.7$\pm$0.2 &79.3$\pm$0.2 
    &47.3$\pm$0.2 &15.6$\pm$0.3 &13.2$\pm$0.3 &40.5$\pm$0.0 
    &51.8$\pm$1.0 &8.1$\pm$1.7 &13.5$\pm$1.5 &29.5$\pm$0.9 \\
    \textbf{AGC-DRR}~(\textit{IJCAI}'22)\cite{2022AGC-DRR-Gong} 
    &80.3$\pm$0.3 &49.9$\pm$0.7 &55.2$\pm$0.7 &79.8$\pm$0.6 
    &43.2$\pm$3.9 &10.7$\pm$2.5 &8.8$\pm$3.1 &31.0$\pm$2.9 
    &44.3$\pm$3.5 &8.3$\pm$2.7 &6.9$\pm$3.4 &29.4$\pm$3.0 \\
    \textbf{CONVERT}~(\textit{MM}'23)\cite{2023CONVERT-yang} 
    &55.4$\pm$2.9 &22.7$\pm$2.7 &20.3$\pm$2.6 &55.4$\pm$3.3 
    &59.9$\pm$2.0 &\underline{35.5$\pm$2.5} &33.0$\pm$2.7 &\underline{47.9$\pm$3.3} 
    &49.7$\pm$4.2 &22.4$\pm$5.2 &18.0$\pm$5.2 &38.4$\pm$4.6 \\
    \textbf{SCGC}~(\textit{TNNLS}'23)\cite{2023SCGC-Liu} 
    &67.1$\pm$1.9 &38.6$\pm$1.7 &35.7$\pm$2.0 &66.9$\pm$2.0 
    &45.9$\pm$4.6 &14.7$\pm$1.9 &11.1$\pm$3.4 &35.3$\pm$2.7 
    &52.3$\pm$1.8 &9.2$\pm$2.0 &14.3$\pm$2.5 &32.5$\pm$1.0 \\
    \textbf{NCLA}~(\textit{AAAI}'23)\cite{2023NCLA-shen-neighbor} 
    &63.9$\pm$0.7 &35.5$\pm$0.7 &26.3$\pm$1.2 &65.0$\pm$0.7 
    &38.2$\pm$3.3 &9.2$\pm$0.7 &5.6$\pm$2.1 &30.7$\pm$2.8 
    &42.0$\pm$0.5 &7.4$\pm$0.5 &5.7$\pm$0.8 &30.1$\pm$0.3 \\
    \textbf{SCAGC}~(\textit{TMM}'23)\cite{2023SCAGC-Xia-TMM} 
    &78.8$\pm$0.1 &48.7$\pm$0.2 &53.4$\pm$0.3 &78.4$\pm$0.1 
    &40.2$\pm$0.8 &12.3$\pm$1.5 &7.2$\pm$1.1 &30.2$\pm$4.4 
    &32.3$\pm$0.3 &9.6$\pm$0.1 &2.7$\pm$0.1 &28.2$\pm$0.4 \\
    \textbf{CCGC}~(\textit{AAAI}'23)\cite{2023CCGC-Yang} 
    &55.6$\pm$1.7 &26.0$\pm$1.8 &19.8$\pm$1.1 &56.1$\pm$1.5 
    &47.9$\pm$3.5 &12.0$\pm$5.4 &9.0$\pm$6.1 &30.4$\pm$9.1 
    &54.9$\pm$1.7 &9.4$\pm$1.7 &15.4$\pm$2.1 &30.8$\pm$2.9 \\
    \textbf{HSAN}~(\textit{AAAI}'23)\cite{2023HSAN-liu} 
    &79.6$\pm$0.3 &49.3$\pm$0.5 &54.2$\pm$0.6 &79.2$\pm$0.3 
    &49.8$\pm$2.9 &14.4$\pm$2.4 &11.7$\pm$2.8 &34.0$\pm$2.5 
    &56.6$\pm$2.1 &11.5$\pm$2.3 &17.7$\pm$3.2 &32.1$\pm$3.6 \\
    \textbf{GraphACL}~(\textit{NIPS}'23)\cite{2023GraphACL-Xiao} 
    &76.0$\pm$0.9 &43.7$\pm$1.2 &47.4$\pm$1.5 &75.5$\pm$1.0 
    &60.7$\pm$1.7 &25.0$\pm$2.7 &27.0$\pm$3.7 &41.6$\pm$2.8 
    &56.7$\pm$1.9 &7.9$\pm$4.0 &7.1$\pm$7.6 &24.1$\pm$7.5 \\
    \textbf{HeterGCL}~(\textit{IJCAI}'24)\cite{2024HeterGCL-Wang} 
    &56.7$\pm$0.3 &40.4$\pm$0.7 &33.2$\pm$0.4 &47.5$\pm$0.8 
    &\underline{65.5$\pm$3.8} &34.0$\pm$6.6 &\underline{34.3$\pm$6.1} &43.3$\pm$3.8 
    &\underline{64.5$\pm$1.3} &\underline{30.7$\pm$3.7} &\underline{30.8$\pm$2.9} &\underline{41.2$\pm$3.2} \\
    \hline
    \textbf{NeuCGC} 
    &\textbf{81.5$\pm$0.2} &\textbf{53.0$\pm$0.3} &\textbf{58.5$\pm$0.4} &\textbf{80.8$\pm$0.2} 
    &\textbf{73.6$\pm$1.5} &\textbf{46.6$\pm$3.9} &\textbf{51.6$\pm$3.2} &\textbf{58.0$\pm$3.2} 
    &\textbf{70.2$\pm$1.5} &\textbf{37.0$\pm$3.0} &\textbf{44.5$\pm$6.0} &\textbf{49.8$\pm$3.1} \\
    \midrule[1.0pt]
    \multirow{2}*{\textbf{Method}} &\multicolumn{4}{c|}{\textbf{Texas~ (0.06)}} &\multicolumn{4}{c|}{\textbf{Chameleon~ (0.23)}} & \multicolumn{4}{c}{\textbf{Average}} \\
     &\textbf{ACC} &\textbf{NMI} &\textbf{ARI} &\textbf{F1} &\textbf{ACC} &\textbf{NMI} &\textbf{ARI} &\textbf{F1} &\textbf{ACC} &\textbf{NMI} &\textbf{ARI} &\textbf{F1}  \\
    \hline
    \textbf{SDCN}~(\textit{WWW}'20)\cite{2020SDCN-bo}
    &58.2$\pm$2.8 &10.4$\pm$4.1 &9.4$\pm$6.7 &24.2$\pm$4.6 
    &34.8$\pm$0.7 &11.5$\pm$1.4 &9.7$\pm$1.3 &27.7$\pm$1.3
    &58.5 &28.7 &28.1 &47.6  \\
    \textbf{DFCN}~(\textit{AAAI}'21)\cite{2021DFCN-tu-deep} 
    &52.4$\pm$1.5 &10.2$\pm$2.5 &18.8$\pm$3.3 &29.0$\pm$1.7
    &35.4$\pm$0.2 &13.7$\pm$0.5 &9.4$\pm$0.3 &30.0$\pm$0.2
    &63.2 &32.7 &34.8 &53.9  \\
    \textbf{AUTOSSL}~(\textit{ICLR}'22)\cite{2022AutoSSL-jin} 
    &65.4$\pm$0.7 &24.9$\pm$1.2 &32.0$\pm$1.5 &39.7$\pm$2.2
    &32.7$\pm$0.2 &9.5$\pm$0.7 &7.5$\pm$0.9 &25.5$\pm$1.2
    &58.2 &31.1 &29.4 &48.6  \\
    \textbf{GraphMAE}~(\textit{KDD}'22)\cite{2022GraphMAE-hou} 
    &68.3$\pm$1.0 &36.3$\pm$2.9 &\underline{44.1$\pm$3.9} &43.5$\pm$3.0 
    &35.8$\pm$0.8 &11.2$\pm$1.8 &7.8$\pm$0.9 &34.9$\pm$0.8
    &66.2 &36.2 &\underline{37.6} &\underline{55.9}  \\
    \textbf{DyFSS}~(\textit{AAAI}'24)\cite{2024DyFSS-zhu} 
    &60.4$\pm$0.5 &15.2$\pm$0.6 &28.3$\pm$0.9 &32.7$\pm$0.6 
    &34.1$\pm$0.2 &11.0$\pm$0.3 &7.7$\pm$0.2 &29.6$\pm$1.2
    &59.8 &28.9 &29.5 &50.4  \\
    \textbf{DGCluster}~(\textit{AAAI}'24)\cite{2024DGCluster-bhowmick} 
    &49.2$\pm$2.8 &9.2$\pm$2.0 &17.6$\pm$4.2 &29.4$\pm$1.7 
    &38.6$\pm$1.2 &18.2$\pm$1.3 &\textbf{15.8$\pm$2.2} &36.2$\pm$1.8
    &60.3 &34.4 &34.4 &52.7  \\
    
    \hline
    \textbf{DCRN}~(\textit{AAAI}'22)\cite{2022DCRN-liu} 
    &51.6$\pm$1.3 &10.3$\pm$1.6 &17.9$\pm$2.9 &29.6$\pm$1.0
    &33.3$\pm$0.2 &9.6$\pm$0.3 &5.5$\pm$0.1 &27.1$\pm$0.1
    &62.6 &33.5 &35.5 &53.7  \\
    \textbf{AGC-DRR}~(\textit{IJCAI}'22)\cite{2022AGC-DRR-Gong} 
    &42.5$\pm$4.0 &5.7$\pm$1.2 &6.6$\pm$3.3 &26.9$\pm$2.0 
    &35.6$\pm$0.3 &14.5$\pm$2.6 &12.5$\pm$2.5 &29.7$\pm$0.6
    &58.6 &31.6 &32.4 &49.3  \\
    \textbf{CONVERT}~(\textit{MM}'23)\cite{2023CONVERT-yang} 
    &49.6$\pm$4.0 &22.0$\pm$5.3 &18.1$\pm$5.1 &37.9$\pm$4.8 
    &35.8$\pm$0.7 &15.9$\pm$1.7 &10.7$\pm$2.4 &33.6$\pm$2.4
    &59.9 &34.0 &32.1 &54.3  \\
    \textbf{SCGC}~(\textit{TNNLS}'23)\cite{2023SCGC-Liu} 
    &52.5$\pm$3.0 &11.5$\pm$1.6 &15.5$\pm$4.4 &28.9$\pm$2.9 
    &35.0$\pm$0.9 &15.3$\pm$2.1 &9.7$\pm$1.9 &31.7$\pm$1.6
    &60.9 &32.1 &32.1 &52.2  \\
    \textbf{NCLA}~(\textit{AAAI}'23)\cite{2023NCLA-shen-neighbor} 
    &46.9$\pm$1.6 &8.0$\pm$0.7 &15.8$\pm$1.7 &29.7$\pm$0.8 
    &34.6$\pm$0.7 &11.1$\pm$1.2 &8.0$\pm$0.8 &34.2$\pm$0.7
    &57.2 &30.4 &29.4 &52.1  \\
    \textbf{SCAGC}~(\textit{TMM}'23)\cite{2023SCAGC-Xia-TMM} 
    &40.0$\pm$0.2 &7.2$\pm$0.3 &6.7$\pm$0.3 &28.3$\pm$0.3 
    &\underline{40.3$\pm$0.3} &11.0$\pm$0.2 &9.3$\pm$0.3 &\underline{40.0$\pm$0.3}
    &56.5 &31.2 &29.6 &52.7  \\
    \textbf{CCGC}~(\textit{AAAI}'23)\cite{2023CCGC-Yang} 
    &55.4$\pm$2.5 &13.8$\pm$2.3 &19.8$\pm$4.0 &29.7$\pm$2.9 
    &36.0$\pm$0.5 &16.2$\pm$1.6 &11.0$\pm$2.3 &33.4$\pm$1.8
    &60.5 &30.7 &30.7 &50.5  \\
    \textbf{HSAN}~(\textit{AAAI}'23)\cite{2023HSAN-liu} 
    &58.0$\pm$2.4 &15.4$\pm$3.6 &21.6$\pm$4.7 &33.7$\pm$4.2 
    &34.8$\pm$0.7 &11.0$\pm$2.8 &7.1$\pm$0.7 &32.4$\pm$1.8
    &64.6 &34.1 &36.2 &55.0  \\
    \textbf{GraphACL}~(\textit{NIPS}'23)\cite{2023GraphACL-Xiao} 
    &\underline{70.3$\pm$1.2} &\underline{38.2$\pm$1.6} &41.6$\pm$4.6 &\underline{44.0$\pm$4.8} 
    &35.8$\pm$0.5 &9.5$\pm$2.1 &6.9$\pm$0.6 &33.4$\pm$0.9
    &\underline{66.4} &36.0 &36.9 &55.2  \\
    \textbf{HeterGCL}~(\textit{IJCAI}'24)\cite{2024HeterGCL-Wang} 
    &64.0$\pm$0.9 &38.5$\pm$1.6 &30.3$\pm$2.6 &33.5$\pm$1.3 
    &40.0$\pm$0.7 &\underline{21.3$\pm$0.3} &14.0$\pm$0.4 &34.5$\pm$0.7 
    &64.0 &\underline{40.1} & 36.9 & 50.6  \\
    \hline
    \textbf{NeuCGC} 
    &\textbf{73.1$\pm$1.7} &\textbf{40.5$\pm$3.2} &\textbf{51.9$\pm$4.5} &\textbf{49.3$\pm$3.8} 
    &\textbf{42.3$\pm$0.4} &\textbf{22.9$\pm$0.1} &\underline{14.8$\pm$0.6} &\textbf{40.4$\pm$2.3}
    &\textbf{72.8} &\textbf{47.4} &\textbf{50.6} &\textbf{63.8}  \\
    \hline \hline
    \end{tabular}
    }
    \label{tb_Comparison}
\end{center}
\end{table*}
Theorem \ref{thm:MI} not only ensures good convergence during model training but also indicates that our method is capable to achieve greater learning performance than the InfoNCE objective by incorporating neutral pairs, thereby yielding better clustering results. In Section \ref{sec_Ana_NeuCL}, we further empirically evaluate the superiority of $\mathcal{L}_{AFC}$ against the InfoNCE objective.

The aforementioned analysis shows that NeuCGC can achieve global and local neighborhood distribution alignment between the learned graph representations, and simultaneously capture better feature consistency. These advantages enable our method to enjoy superior clustering performance.

\section{Experiments}
\subsection{Experimental Setup}
\noindent \textbf{Datasets.}
We evaluate our method using eleven benchmark graph datasets, including six homophilic graph datasets: citation networks Cora, Citeseer, and Pubmed \cite{2008Cora-esen-Collective}, paper network ACM \cite{2022survey-liu}, author network DBLP \cite{2022survey-liu}, and Amazon co-purchase graph Photo \cite{2015Amazon-mcauley}, along with five heterophilic graph datasets: webpage networks Texas, Wisconsin and Cornell \cite{2009Texas-tang}, as well as two Wikipedia webpage networks Chameleon and Crocodile \cite{2021Chameleon-Benedek}. 
Among them, Photo is a moderate scale dataset, Pubmed and Crocodile are large scale datasets, and the others are small datasets.
The range of homophily ratios across them offers a comprehensive foundation for evaluating the effectiveness of our proposed NeuCGC. Detailed statistics for these datasets are provided in Table \ref{tb_Statistics}.

\noindent \textbf{Competing methods.} 
We compare the clustering performance of our approach with sixteen state-of-the-art graph clustering methods, including classical methods: 
SDCN \cite{2020SDCN-bo}, DFCN \cite{2021DFCN-tu-deep}, AutoSSL \cite{2022AutoSSL-jin}, GraphMAE \cite{2022GraphMAE-hou}, DyFSS \cite{2024DyFSS-zhu}, and DGCluster \cite{2024DGCluster-bhowmick}, 
and contrastive methods: 
DCRN \cite{2022DCRN-liu}, AGC-DRR \cite{2022AGC-DRR-Gong}, CONVERT \cite{2023CONVERT-yang}, SCGC \cite{2023SCGC-Liu}, NCLA \cite{2023NCLA-shen-neighbor}, SCAGC \cite{2023SCAGC-Xia-TMM}, CCGC \cite{2023CCGC-Yang}, HSAN \cite{2023HSAN-liu}, GraphACL \cite{2023GraphACL-Xiao}, and HeterGCL \cite{2024HeterGCL-Wang}.

\noindent \textbf{Metrics.}
We evaluate the quality of cluster assignment using four widely used metrics \cite{2022survey-liu, 2020SDCN-bo, 2023SGC-Liu, 2024LaSA-Liu, 2024DyFSS-zhu}, namely, Clustering Accuracy (ACC), Normalized Mutual Information (NMI), Adjusted Rand Index (ARI), and Macro-F1 Score (F1). A higher value indicates better clustering performance. 

\noindent \textbf{Implementation Details.}
All experiments are conducted on the PyTorch platform and NVIDIA GPUs, with Pubmed and Crocodile processed on an RTX 4090 (24GB), following existing methods \cite{2022DCRN-liu},\cite{2023SCGC-Liu},\cite{2024DyFSS-zhu},\cite{2024HeterGCL-Wang}, and others run on an RTX 2080 Super (8GB) GPU. 
For baselines, we use the official implementations provided by the authors. A grid search is performed to select the optimal hyper-parameters for every method, and all experiments are repeated 10 times with different random seeds. Our model is trained for 500 epochs until convergence with the Adam optimizer \cite{2017Adam-Kingma}. The dimension $d$ is 300 for DBLP, while empirically set to 1000 for other datasets; the learning rate is selected from \{1e-03, 1e-04, 1e-05\}; the trade-off factors $\lambda_1$ and $\lambda_2$ are selected from \{0.01, 0.1, 0.5, 1, 5, 10, 100\}; the confidence signal $k$ is tuned from 0.1 to 1.0 with interval 0.1.

\subsection{Comparison Results}
We first report the comparison results in Table \ref{tb_Comparison} and Tabel \ref{tb_Large_scale}.
Building on these results, we summarize our key observations below:
\begin{itemize}
    \item Compared to classical DGC algorithms such as SDCN and DFCN, our method and other contrastive graph clustering methods generally achieve better performance. Even when compared to the latest DyFSS and DGCluster, our approach exhibits a significant advantage. This is because the contrastive mechanism encourages the network to capture more self-supervised information, enabling it to learn more discriminative representations and achieve higher clustering accuracy.
    \item NeuCGC significantly outperforms traditional CGC methods such as DCRN, AGC-DRR, and CONVERT, as well as neighbor-based CGC approaches like SCGC and NCLA, particularly on graphs with low homophily ratios. For instance, on the Wisconsin and Texas datasets, it achieves an over 20\% improvement in ACC compared to SCGC. This demonstrates the effectiveness of the proposed neutral contrastive learning mechanism in addressing low-homophily graphs. Moreover, it highlights the effectiveness of our method when faced with graphs with different homophily ratios.
    \item Further, even on larger-scale datasets (Table \ref{tb_Large_scale}), NeuCGC demonstrates competitive clustering performance, underscoring its robustness and scalability. On widely used homophilic graph datasets, e.g. Photo and Pubmed, it achieves leading performance across multiple evaluation metrics. Similarly, on the heterophilic graph dataset Crocodile, it outperforms nearly all competing baselines. These results, spanning datasets of varying sizes and homophily levels, provide comprehensive evidence of the superiority and versatility of our method.
    \item Overall, our approach achieves superior or comparable performance compared to the latest DGC baseline methods, e.g., DyFSS, DGCluster, HSAN, GraphACL, and HeterGCL, regardless of the level of graph homophily. Particularly on the Wisconsin and Cornell datasets, it outperforms the runner-up by at least 5.7\% across all metrics. When considering the average scores, it maintains a significant lead, with margins going up to 6.4\%, 7.3\%, 13\% and 7.9\% in the ACC, NMI, ARI and F1 metrics, respectively. These results demonstrate that NeuCGC, leveraging global and local neighborhood distribution alignment alongside feature consistency neutral contrastive learning, enables efficient representation learning and achieves exceptional clustering performance on graphs with varying homophily levels.
\end{itemize} 
\begin{table*}[!t]
    \renewcommand{\arraystretch}{1.15}
    \caption{Comparison of average clustering performance under 10 runs on three moderate and large scale graph datasets (as well as their homophily ratios $r_h$). Four metrics (in \%) are used to evaluate the clustering results. ``OOM'' indicates running out-of-memory on 24 GB GPU memory. The best and second best results are marked with boldface and underline, respectively.}
    \begin{center}
    \setlength{\tabcolsep}{0.7mm}{
    \begin{tabular}{c|cccc|cccc|cccc}
    \hline \hline 
    \multirow{2}*{\textbf{Method}} & \multicolumn{4}{c|}{\textbf{Photo~(0.83)}} & \multicolumn{4}{c|}{\textbf{Pubmed~ (0.80)}} & \multicolumn{4}{c}{\textbf{Crocodile~(0.24)}} \\
     &\textbf{ACC} &\textbf{NMI} &\textbf{ARI} &\textbf{F1} &\textbf{ACC} &\textbf{NMI} &\textbf{ARI} &\textbf{F1} &\textbf{ACC} &\textbf{NMI} &\textbf{ARI} &\textbf{F1} \\
    \hline
    \textbf{SDCN}~(\textit{WWW}'20)\cite{2020SDCN-bo}
    &57.1$\pm$2.0 &45.4$\pm$2.5 &38.0$\pm$2.0 &41.5$\pm$4.0 &63.9$\pm$0.8 &25.1$\pm$2.1 &23.7$\pm$2.1 &64.8$\pm$0.9 &41.5$\pm$1.1 &14.1$\pm$2.0 &12.3$\pm$1.0 &36.8$\pm$1.9 \\
    \textbf{DFCN}~(\textit{AAAI}'21)\cite{2021DFCN-tu-deep} 
    &77.3$\pm$0.1 &68.8$\pm$0.4 &60.4$\pm$0.3 &72.0$\pm$0.1 &68.0$\pm$0.1 &32.5$\pm$0.2 &30.4$\pm$0.2 &67.4$\pm$0.1 &43.4$\pm$1.1 &16.7$\pm$0.8 &13.3$\pm$0.6 &40.5$\pm$2.0 \\
    \textbf{AUTOSSL}~(\textit{ICLR}'22)\cite{2022AutoSSL-jin} 
    &52.4$\pm$1.2 &44.6$\pm$2.0 &23.2$\pm$1.6 &47.7$\pm$1.7 &65.2$\pm$0.7 &29.4$\pm$1.4 &26.8$\pm$1.1 &65.1$\pm$0.5 &35.1$\pm$0.1 &5.3$\pm$0.5 &6.4$\pm$0.4 &24.2$\pm$0.5 \\
    \textbf{GraphMAE}~(\textit{KDD}'22)\cite{2022GraphMAE-hou} 
    &76.6$\pm$3.3 &64.9$\pm$3.7 &56.8$\pm$4.4 &73.1$\pm$4.1 &69.2$\pm$0.6 &31.6$\pm$1.0 &31.4$\pm$0.6 &68.9$\pm$0.8 &44.0$\pm$1.7 &17.5$\pm$0.8 &13.7$\pm$1.8 &41.4$\pm$1.8 \\
    \textbf{DyFSS}~(\textit{AAAI}'24)\cite{2024DyFSS-zhu} 
    &49.6$\pm$0.4 &46.7$\pm$0.6 &25.5$\pm$1.7 &46.8$\pm$1.0 &66.7$\pm$2.3 &27.7$\pm$4.4 &27.1$\pm$4.1 &66.9$\pm$1.7 &41.4$\pm$1.5 &15.0$\pm$1.5 &11.9$\pm$1.3 &33.6$\pm$2.3 \\
    \textbf{DGCluster}~(\textit{AAAI}'24)\cite{2024DGCluster-bhowmick} 
    &78.3$\pm$1.8 &68.3$\pm$2.6 &58.4$\pm$2.0 &\underline{76.1$\pm$1.5} &67.3$\pm$4.0 &28.3$\pm$3.5 &28.2$\pm$4.1 &65.9$\pm$6.3 &39.2$\pm$2.5 &11.5$\pm$5.7 &8.8$\pm$3.6 &31.5$\pm$4.4 \\
    \hline
    \textbf{DCRN}~(\textit{AAAI}'22)\cite{2022DCRN-liu} 
    &\underline{79.7$\pm$0.1} &\textbf{73.1$\pm$0.4} &\underline{63.2$\pm$0.3} &73.7$\pm$0.2 &69.5$\pm$0.2 &32.1$\pm$0.3 &31.6$\pm$0.3 &68.6$\pm$0.2 &43.1$\pm$0.5 &19.1$\pm$0.8 &14.2$\pm$0.2 &39.7$\pm$0.1 \\
    \textbf{AGC-DRR}~(\textit{IJCAI}'22)\cite{2022AGC-DRR-Gong} 
     &77.3$\pm$3.7 &\underline{71.7$\pm$1.7} &61.7$\pm$3.2 &71.8$\pm$3.7 &66.9$\pm$2.2 &27.3$\pm$3.0 &26.9$\pm$3.4 &67.3$\pm$1.8 &42.5$\pm$0.9 &14.3$\pm$0.5 &14.1$\pm$0.3 &39.6$\pm$0.9 \\
     \textbf{CONVERT}~(\textit{MM}'23)\cite{2023CONVERT-yang} 
     &77.1$\pm$0.5 &67.2$\pm$1.0 &60.7$\pm$1.8 &74.0$\pm$1.0 &68.0$\pm$1.8 &29.8$\pm$2.6 &29.6$\pm$2.7 &67.7$\pm$1.6 &44.2$\pm$1.0 &16.8$\pm$0.8 &\textbf{15.2$\pm$1.1} &42.4$\pm$1.1 \\
    \textbf{SCGC}~(\textit{TNNLS}'23)\cite{2023SCGC-Liu} 
     &77.4$\pm$0.3 &67.6$\pm$0.8 &58.4$\pm$0.7 &72.2$\pm$0.9 &63.7$\pm$1.3 &30.0$\pm$0.9 &27.5$\pm$0.7 &63.0$\pm$1.6 &43.4$\pm$0.7 &16.8$\pm$0.8 &13.6$\pm$0.9 &40.3$\pm$1.8 \\
     \textbf{NCLA}~(\textit{AAAI}'23)\cite{2023NCLA-shen-neighbor} 
     &78.3$\pm$0.5 &69.7$\pm$0.6 &60.9$\pm$0.6 &72.6$\pm$0.2 &67.5$\pm$0.2 &32.7$\pm$0.3 &29.6$\pm$0.4 &66.9$\pm$0.2 &38.8$\pm$0.7 &15.7$\pm$0.2 &12.4$\pm$0.5 &34.0$\pm$1.3  \\
    \textbf{SCAGC}~(\textit{TMM}'23)\cite{2023SCAGC-Xia-TMM} 
     &76.6$\pm$0.6 &65.6$\pm$0.6 &58.6$\pm$1.3 &71.1$\pm$0.8 &\textbf{71.9$\pm$0.2} &\underline{34.9$\pm$0.2} &\underline{33.5$\pm$0.3} &\textbf{71.0$\pm$0.2} &42.1$\pm$0.1 &19.9$\pm$0.3 &13.7$\pm$0.1 &40.0$\pm$0.1 \\
    \textbf{CCGC}~(\textit{AAAI}'23)\cite{2023CCGC-Yang} 
     &77.2$\pm$0.4 &67.4$\pm$0.4 &57.9$\pm$0.6 &72.1$\pm$0.5 &63.6$\pm$1.6 &31.6$\pm$0.8 &28.7$\pm$0.7 &63.0$\pm$2.0 &\underline{44.5$\pm$1.2} &18.0$\pm$2.2 &14.4$\pm$0.6 &\underline{42.5$\pm$2.3}  \\
    \textbf{HSAN}~(\textit{AAAI}'23)\cite{2023HSAN-liu} 
     &77.3$\pm$0.3 &67.1$\pm$0.3 &58.0$\pm$0.5 &71.8$\pm$0.2 
     &\multicolumn{4}{c|}{OOM} &\multicolumn{4}{c}{OOM}\\
    \textbf{GraphACL}~(\textit{NIPS}'23)\cite{2023GraphACL-Xiao} 
     &78.1$\pm$0.8 &66.8$\pm$0.9 &58.9$\pm$1.0 &75.1$\pm$3.1 &67.0$\pm$1.1 &30.1$\pm$2.3 &28.7$\pm$2.0 &66.9$\pm$1.0 &44.0$\pm$1.5 &15.9$\pm$1.4 &14.1$\pm$0.6 &41.3$\pm$2.9 \\ 
     \textbf{HeterGCL}~(\textit{IJCAI}'24)\cite{2024HeterGCL-Wang} 
     &70.1$\pm$1.2 &61.0$\pm$0.4 &50.7$\pm$1.0 &68.8$\pm$1.7 &61.5$\pm$1.5 &28.7$\pm$1.6 &25.4$\pm$1.6 &60.6$\pm$2.1 &40.5$\pm$0.9 &\underline{20.4$\pm$1.5} &7.4$\pm$0.5 &39.2$\pm$1.3 \\
     \hline
    \textbf{NeuCGC}
     &\textbf{81.4$\pm$0.3} &70.8$\pm$0.6 &\textbf{63.4$\pm$0.8} &\textbf{80.1$\pm$0.5} &\underline{70.7$\pm$0.1} &\textbf{36.1$\pm$0.2} &\textbf{34.3$\pm$0.2} &\underline{70.3$\pm$0.1} &\textbf{46.2$\pm$0.9} &\textbf{21.4$\pm$1.5} &\underline{14.5$\pm$1.1} &\textbf{45.4$\pm$1.2} \\
    \hline \hline
    \end{tabular}
    }
    \label{tb_Large_scale}
\end{center}
\vspace{-1.0em}
\end{table*}
\begin{table}[!t]
    \renewcommand\arraystretch{1.2}
    \caption{Ablation study of each component.}
    \centering{
    \setlength{\tabcolsep}{1.0mm}{
    \begin{tabular}{c c| cccc}
        \hline \hline 
        \textbf{Dataset} &\textbf{Metric} &\textbf{w/o \textbf{AFC}} &\textbf{w/o \textbf{GDA}} &\textbf{w/o \textbf{NCA}} &\textbf{NeuCGC} \\
        \hline
        \multirow{4}{*}{\textbf{Cora}}
        &ACC &73.6$\pm$0.5 &74.8$\pm$1.1 &75.4$\pm$0.8 &\textbf{77.1$\pm$1.0} \\
        &NMI &56.0$\pm$0.6 &56.6$\pm$1.1 &56.8$\pm$0.8 &\textbf{59.0$\pm$1.1} \\
        &ARI &50.9$\pm$0.8 &53.3$\pm$1.7 &54.1$\pm$1.0 &\textbf{56.3$\pm$1.4} \\
        &F1 &69.3$\pm$1.2 &72.3$\pm$2.2 &74.0$\pm$1.1 &\textbf{75.8$\pm$1.2} \\
        \hline
        \multirow{4}{*}{\textbf{DBLP}}
        &ACC &76.4$\pm$0.4 &70.3$\pm$1.9 &80.9$\pm$0.1	 &\textbf{81.5$\pm$0.2} \\
        &NMI &46.0$\pm$0.5 &39.7$\pm$1.4 &52.0$\pm$0.2	 &\textbf{53.0$\pm$0.3}  \\
        &ARI &48.4$\pm$0.7 &38.1$\pm$2.1 &57.1$\pm$0.2	 &\textbf{58.5$\pm$0.4}  \\
        &F1 &76.1$\pm$0.4 &70.4$\pm$1.9 &80.4$\pm$0.1	 &\textbf{80.8$\pm$0.2} \\
        \hline
        \multirow{4}{*}{\textbf{Wisconsin}}
        &ACC &71.8$\pm$1.3 &73.1$\pm$1.3 &73.6$\pm$1.5 &\textbf{73.6$\pm$1.5}  \\
        &NMI &45.6$\pm$3.1 &46.0$\pm$3.7 &46.2$\pm$4.1 &\textbf{46.6$\pm$3.9} \\
        &ARI &49.0$\pm$2.4 &51.4$\pm$2.4 &51.3$\pm$3.6 &\textbf{51.6$\pm$3.2}  \\
        &F1 &54.5$\pm$3.6 &56.0$\pm$2.8 &57.8$\pm$3.5 &\textbf{58.0$\pm$3.2}  \\
        \hline
        \multirow{4}{*}{\textbf{Cornell}}
        &ACC &70.0$\pm$1.5 &69.7$\pm$1.6 &70.1$\pm$1.9	 &\textbf{70.2$\pm$1.5}  \\
        &NMI &36.7$\pm$3.0 &35.7$\pm$3.0 &36.2$\pm$2.7	 &\textbf{37.0$\pm$3.0}  \\
        &ARI &44.3$\pm$5.1 &43.9$\pm$5.7 &43.3$\pm$3.2	 &\textbf{44.5$\pm$6.0}  \\
        &F1 &48.7$\pm$2.3 &48.2$\pm$3.9 &48.6$\pm$3.9	 &\textbf{49.8$\pm$3.1}  \\
        \hline \hline
    \end{tabular}
    }}
    \label{tb_Ablation}
    \vspace{-0.5em}
\end{table}
\begin{figure}[!t]
    \centering{
        \includegraphics[width=0.48\textwidth]{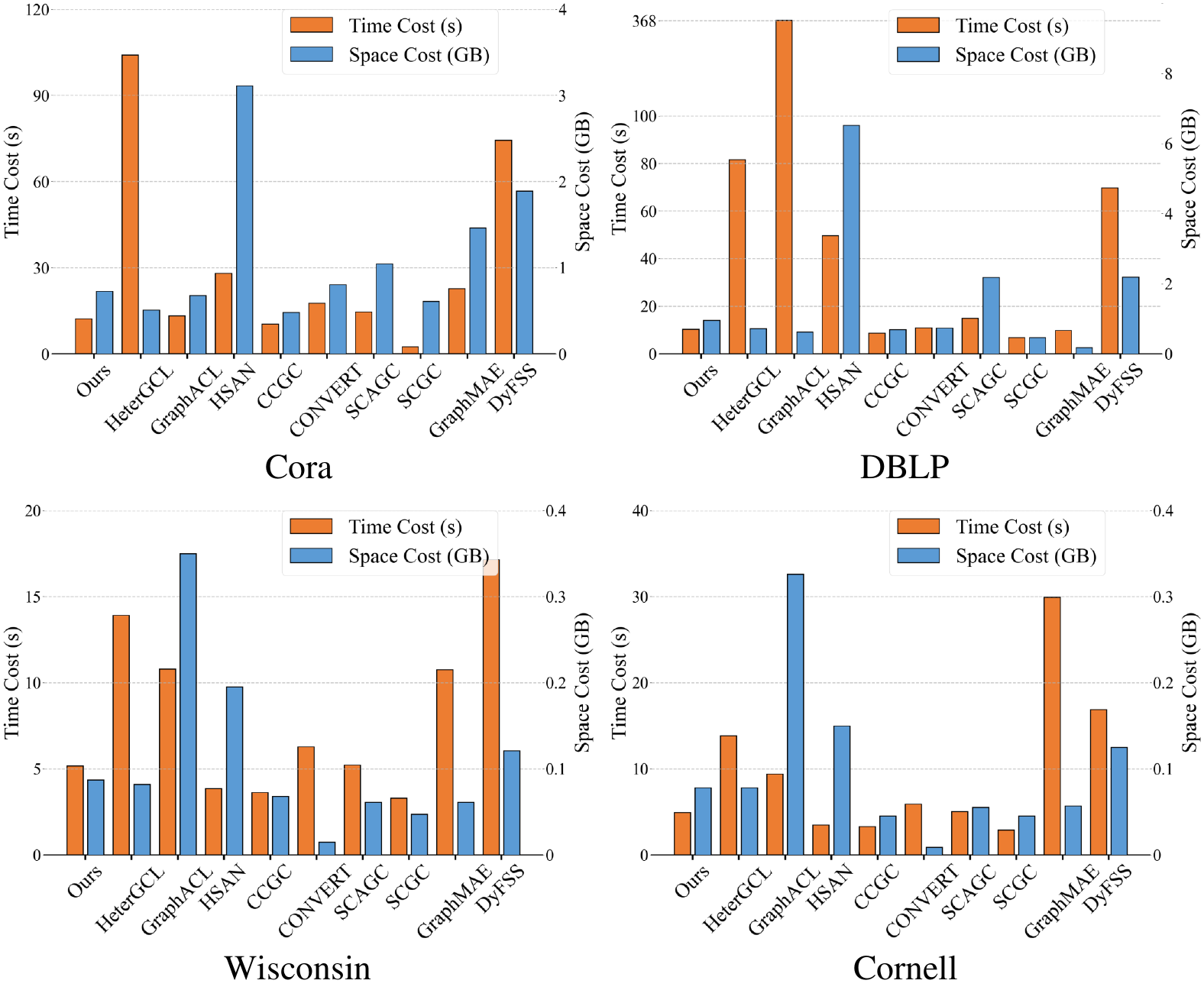}
        \captionsetup{justification=centering}
        \caption{Computational cost comparison of our NeuCGC against seven CGC methods and two conventional DGC methods on four datasets.}
    \label{fig_complexity}
    }
    \vspace{-1.5em}
\end{figure}
\begin{table}[!t]
    \renewcommand\arraystretch{1.2}
    \caption{Ablation studies of the neutral contrastive factor $\eta$ and the learned high-confidence graph $\textbf{H}$.}
    \centering{
    \setlength{\tabcolsep}{2.0mm}{
    \begin{tabular}{c c| ccc}
        \hline \hline 
        \textbf{Dataset} &\textbf{Metric} &\textbf{w/o $\eta$} &\textbf{InfoNCE} &\textbf{NeuCGC} \\
        \hline
        \multirow{4}{*}{\textbf{Cora}}
        &ACC &76.4$\pm$0.9 &75.0$\pm$0.9 &\textbf{77.1$\pm$1.0} \\
        &NMI &58.1$\pm$1.3 &55.8$\pm$0.8 &\textbf{59.0$\pm$1.1} \\
        &ARI &55.7$\pm$1.5 &53.1$\pm$1.6 &\textbf{56.3$\pm$1.4} \\
        &F1 &74.5$\pm$1.5 &73.5$\pm$1.0 &\textbf{75.8$\pm$1.2} \\
        \hline
        \multirow{4}{*}{\textbf{DBLP}}
        &ACC &80.8$\pm$0.2 &76.6$\pm$0.8 &\textbf{81.5$\pm$0.2} \\
        &NMI &51.8$\pm$0.3 &45.0$\pm$1.1 &\textbf{53.0$\pm$0.3}  \\
        &ARI &56.8$\pm$0.4 &48.1$\pm$1.5 &\textbf{58.5$\pm$0.4}  \\
        &F1 &80.3$\pm$0.2 &76.3$\pm$0.7 &\textbf{80.8$\pm$0.2} \\
        \hline
        \multirow{4}{*}{\textbf{Wisconsin}}
        &ACC &73.5$\pm$1.3 &72.4$\pm$1.5 &\textbf{73.6$\pm$1.5}  \\
        &NMI &46.3$\pm$3.8 &46.1$\pm$3.3 &\textbf{46.6$\pm$3.9} \\
        &ARI &50.9$\pm$3.3 &49.6$\pm$2.9 &\textbf{51.6$\pm$3.2}  \\
        &F1 &57.7$\pm$3.6 &54.3$\pm$4.6 &\textbf{58.0$\pm$3.2}  \\
        \hline
        \multirow{4}{*}{\textbf{Cornell}}
        &ACC &69.4$\pm$1.5 &68.6$\pm$1.5 &\textbf{70.2$\pm$1.5}  \\
        &NMI &36.9$\pm$2.0 &33.8$\pm$3.2 &\textbf{37.0$\pm$3.0}  \\
        &ARI &41.8$\pm$4.1 &41.1$\pm$4.6 &\textbf{44.5$\pm$6.0}  \\
        &F1 &49.1$\pm$2.9 &46.4$\pm$3.5 &\textbf{49.8$\pm$3.1}  \\
        \hline \hline
    \end{tabular}
    }}
    \label{tb_Ablation_eta_H}
    \vspace{-1.0em}
\end{table}
\begin{figure*}[!htb]
    \centering{
        \includegraphics[width=\textwidth]{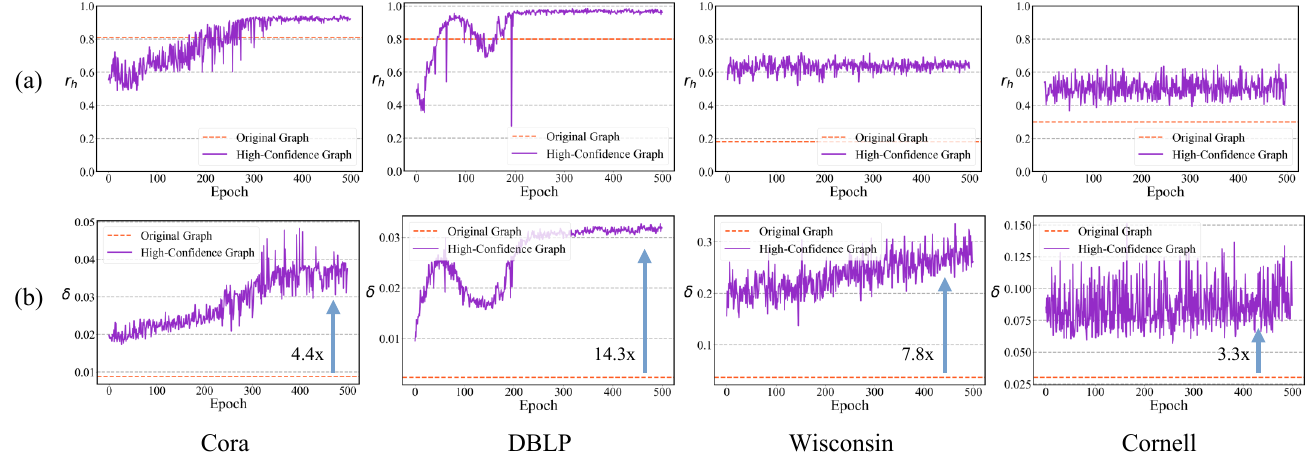}
        \captionsetup{justification=centering}
        \caption{Training curves of (a) homophily ratios $r_h$ and (b) graph neighborhood congener ratios $\delta$ in the original graph and the learned high-confidence graph $\textbf{H}$, respectively. It is evident that on the Cora, DBLP, Wisconsin, and Cornell datasets, $\textbf{H}$ improves $\delta$ by approximately 4.4, 14.3, 7.8, and 3.3 times, respectively, compared with the original graph. Meanwhile, the homophily of $\textbf{H}$ also improves across all datasets. This indicates that the quality of neighbors in $\textbf{H}$ is superior to that in the original graph.}
    \label{fig_H}
    }
    \vspace{-1.5em}
\end{figure*}
\begin{table*}[!h]
    \renewcommand{\arraystretch}{1.2}
    \caption{The clustering performance of migrating the NCA and AFC modules to the traditional CGC method CONVERT \cite{2023CONVERT-yang} and neighbor-based CGC method SCGC \cite{2023SCGC-Liu}.}
    \begin{center}
    \setlength{\tabcolsep}{0.8mm}{
    \begin{tabular}{c|cccc|cccc}
    \hline \hline 
    \multirow{2}*{\textbf{Method}} & \multicolumn{4}{c|}{\textbf{Cora}} & \multicolumn{4}{c}{\textbf{DBLP}}  \\
     &\textbf{ACC} &\textbf{NMI} &\textbf{ARI} &\textbf{F1} &\textbf{ACC} &\textbf{NMI} &\textbf{ARI} &\textbf{F1} \\
    \hline
    \textbf{CONVERT}
    &73.9$\pm$1.0 &55.6$\pm$1.0 &50.2$\pm$1.5 &72.1$\pm$2.8
    &55.4$\pm$2.9 &22.7$\pm$2.7 &20.3$\pm$2.6 &55.4$\pm$3.3  \\
    \textbf{CONVERT+NeuCL} 
    &74.8$\pm$1.2(\textbf{$\uparrow $0.9}) &56.0$\pm$0.9(\textbf{$\uparrow $0.4}) &52.4$\pm$2.5(\textbf{$\uparrow $2.2}) &73.6$\pm$1.3(\textbf{$\uparrow $1.5})
    &58.2$\pm$3.4(\textbf{$\uparrow $2.8}) &24.5$\pm$2.5(\textbf{$\uparrow $1.8}) &22.3$\pm$2.5(\textbf{$\uparrow $2.0}) &58.5$\pm$3.3(\textbf{$\uparrow $3.1}) \\
    \hline
    \textbf{SCGC}
    &73.2$\pm$1.5 &55.6$\pm$0.9 &51.5$\pm$1.9 &70.3$\pm$2.1
    &67.1$\pm$1.9 &38.6$\pm$1.7 &35.7$\pm$2.0 &66.9$\pm$2.0 \\
    \textbf{SCGC+NeuCL} 
    &74.1$\pm$0.7(\textbf{$\uparrow $0.9}) &56.2$\pm$0.9(\textbf{$\uparrow $0.6}) &51.8$\pm$0.6(\textbf{$\uparrow $0.3}) &72.8$\pm$0.9(\textbf{$\uparrow $2.5})
    &74.4$\pm$1.1(\textbf{$\uparrow $7.3}) &44.2$\pm$1.0(\textbf{$\uparrow $5.6}) &44.4$\pm$2.2(\textbf{$\uparrow $8.7}) &74.4$\pm$1.0(\textbf{$\uparrow $7.5}) \\
    
    \midrule[1.0pt]
    \multirow{2}*{\textbf{Method}} & \multicolumn{4}{c|}{\textbf{Wisconsin}} & \multicolumn{4}{c}{\textbf{Cornell}}  \\
     &\textbf{ACC} &\textbf{NMI} &\textbf{ARI} &\textbf{F1} &\textbf{ACC} &\textbf{NMI} &\textbf{ARI} &\textbf{F1} \\
    \hline
    \textbf{CONVERT}
    &59.9$\pm$2.0 &35.5$\pm$2.5 &33.0$\pm$2.7 &47.9$\pm$3.3
    &49.7$\pm$4.2 &22.4$\pm$5.2 &18.0$\pm$5.2 &38.4$\pm$4.6  \\
    \textbf{CONVERT+NeuCL} 
    &61.2$\pm$2.0(\textbf{$\uparrow $1.3}) &35.3$\pm$3.6($\downarrow $0.2) &33.9$\pm$3.0(\textbf{$\uparrow $0.9}) &49.3$\pm$3.2(\textbf{$\uparrow $1.4}) 
    &52.8$\pm$3.5(\textbf{$\uparrow $3.1}) &24.4$\pm$2.9(\textbf{$\uparrow $2.0}) &20.7$\pm$4.2(\textbf{$\uparrow $2.7}) &40.5$\pm$2.4(\textbf{$\uparrow $2.1}) \\
    \hline
    \textbf{SCGC}
    &45.9$\pm$4.6 &14.7$\pm$1.9 &11.1$\pm$3.4 &35.3$\pm$2.7 &52.3$\pm$1.8 &9.2$\pm$2.0 &14.3$\pm$2.5 &32.5$\pm$1.0 \\
    \textbf{SCGC+NeuCL}     
    &50.5$\pm$2.4(\textbf{$\uparrow $4.6}) &15.7$\pm$1.6(\textbf{$\uparrow $1.0}) &12.6$\pm$2.2(\textbf{$\uparrow $1.5}) &36.6$\pm$2.7(\textbf{$\uparrow $1.3}) &56.9$\pm$1.0(\textbf{$\uparrow $4.6}) &11.1$\pm$1.6(\textbf{$\uparrow $1.9}) &17.2$\pm$1.9(\textbf{$\uparrow $2.9}) &33.6$\pm$3.6(\textbf{$\uparrow $1.1}) \\
    \hline \hline
    \end{tabular}
    }
    \label{tb_Transfer}
\end{center}
\vspace{-1.5em}
\end{table*}
\subsection{Ablation Study on Each Component}
In this subsection, we conduct an ablation study to evaluate the effectiveness of each component for our proposed method. The study includes three variants of our method: removing the AFC module (w/o AFC), removing the GDA module (w/o GDA), removing the NCA module (w/o NCA), as well as the full model.
We compare the clustering results of the proposed method and its variants on several graph datasets, as shown in Table \ref{tb_Ablation}. 
The results reveal the following observations: 
\begin{itemize}
    \item The two distribution alignment modules, GDA and NCA, effectively reinforce the consistency of semantic information between two distinct feature distributions from global and local perspectives, respectively. This design enhances the model's clustering performance to varying degrees. 
    \item By identifying high-confidence neighbors for neutral contrastive learning to enhance feature consistency, the AFC module significantly improves the clustering quality on both homophilic and heterophilic graph datasets.
\end{itemize}

The results demonstrate that our full model consistently outperforms its variants across all settings. This indicates that the proposed GDA, NCA, and AFC modules jointly contribute to the overall clustering performance, highlighting the effectiveness of our learning scheme on graphs with homophily variation.

\subsection{Computational Efficiency Analysis}
Theoretical analysis in Section \ref{sec_optimization} proves NeuCGC shares the same complexity order as InfoNCE loss, incurring no significant increase in time and space complexity. To empirically validate this analysis, we evaluate the training time cost (in seconds, s) and GPU memory usage (in GB) on a single NVIDIA RTX 2080 Super (8GB) GPU of our method against seven state-of-the-art CGC methods and two traditional DGC approaches across four datasets. As displayed in Fig. \ref{fig_complexity} using histogram, the experimental results clearly show that NeuCGC achieves comparable and practical training complexity consumption to the most efficient baselines, confirming that it achieves a reasonable balance between performance and computational efficiency. This alignment between theoretical and empirical complexity, combined with its leading clustering performance, demonstrates our method as both computationally efficient and representationally powerful for diverse graph homophily levels.

\subsection{Analysis of The Proposed Neutral Contrastive Learning}\label{sec_Ana_NeuCL}
In this subsection, we conduct several experiments to  to thoroughly analyze the effectiveness of our proposed neutral contrastive learning mechanism from multiple perspectives:

\textit{1) Effectiveness of the neutral contrastive factor $\eta$ and the learned high-confidence graph.}
The neutral contrastive factor and learned graph $\textbf{H}$ are two key components in our proposed method. In \textcolor{blue}{(\ref{eq_L_NCA})}, $\eta$ is a coarse-grained neutral contrastive factor and determines the extent that neutral pairs are considered as positive in NCA module. In terms of $\textbf{H}$, it not only expands the neutral pairs but also learns a specific fine-grained neutral contrastive factor for each neutral pair. As shown in Table \ref{tb_Ablation_eta_H}, we record the clustering results of our method and other two variants, namely, removing $\eta$ and neutral pairs from NCA module (w/o $\eta$), and degenerating into InfoNCE objective via removing $\textbf{H}$ in AFC module (InfoNCE). From the table, we observe the following: 
\begin{itemize}
    \item The absence of coarse-grained neutral contrastive factor $\eta$ and neutral pairs causes our method to produce only suboptimal solution, especially on homophilic graph datasets. Regarding $\eta$, we provide more in-depth analysis in the Supplementary Material.
    \item When $\textbf{H}$ is abandoned, our $\mathcal{L}_{AFC}$ objective degenerates into the traditional InfoNCE objective. In contrast, our method, with the assistance of high-quality neutral pairs, achieves significant improvements in clustering performance on both homophilic and heterophilic graphs. 
\end{itemize}

The above results collectively demonstrate the effectiveness of the proposed neutral contrastive learning mechanism.

\textit{2) The quality of the high-confidence graph $\mathbf{H}$.}
In AFC module, we search new congener nodes and incorporating them into new neighborhoods to obtain $\textbf{H}$. These nodes are then used to form neutral pairs along with the central node, thereby promoting the model to learn richer neighborhood consistency information. However, during practical model training, the newly discovered nodes may either be congeners or other semantically inconsistent nodes which could misguide the representation learning. Here, it is necessary to analyze the quality of the discovered neighbors. We primarily focus on two graph properties: homophily ratio $r_h$ and graph neighborhood congener ratio $\delta$. Our goal is to ensure that more congeners are added as new neighbors and simultaneously these neighbors do not degrade the graph’s homophily. Specifically, let $r_{h0}$ and $\delta_0$ represent the homophily ratio and graph neighborhood congener ratio of the original graph, respectively, while $r_h$ and $\delta$ represent those of $\textbf{H}$. According to \textcolor{blue}{(\ref{eq_cmp_H})}, the inequality $\delta \ge \delta_0$ holds. In this case, if $r_h > r_{h0}$, it indicates that more congeners are present among new neighbors, and we can conclude that $\textbf{H}$ is of higher quality. Conversely, if $r_h \le r_{h0}$, then there are more non-congeners, and the quality of neighbors in $\textbf{H}$ is poorer.

As shown in Fig. \ref{fig_H}, we record the homophily ratio $r_h$ and graph neighborhood congener ratio $\delta$ of $\textbf{H}$ at each epoch during model training. Compared to the original graph, $\textbf{H}$ has higher $r_h$ and $\delta$, indicating that all or most of the newly added neighbors are of the same category. This results in an expected higher-quality graph, which mitigates the negative impact of low homophily and facilitates the model’s feature consistency learning, thereby improving the clustering performance.

\textit{3) Improving traditional CGC method and neighbor-based CGC  method.}
To further investigate the effectiveness of our proposed neutral contrastive learning, we evaluate the clustering performance of two varients (CONVERT+NeuCL and SCGC+NeuCL) by jointly migrating the proposed NCA module and AFC module to traditional CGC method CONVERT \cite{2023CONVERT-yang} and homophily assumption-based neighbor CGC method SCGC \cite{2023SCGC-Liu}, respectively. 
As shown in Table \ref{tb_Transfer}, in comparison to their initial clustering results, our neutral contrastive learning achieves considerable enhancements across both homophilic and heterophilic graph datasets in a plug-and-play manner. This is because our modules enable those models to capture more reliable and discriminative features from neighborhoods, leading to improved clustering performance. 
From these experiments, it is evident that our method not only produces superior clustering results but also has the potential to enhance other CGC approaches.
\begin{figure*}[!t]
    \centering{
        \includegraphics[width=1.0\textwidth]{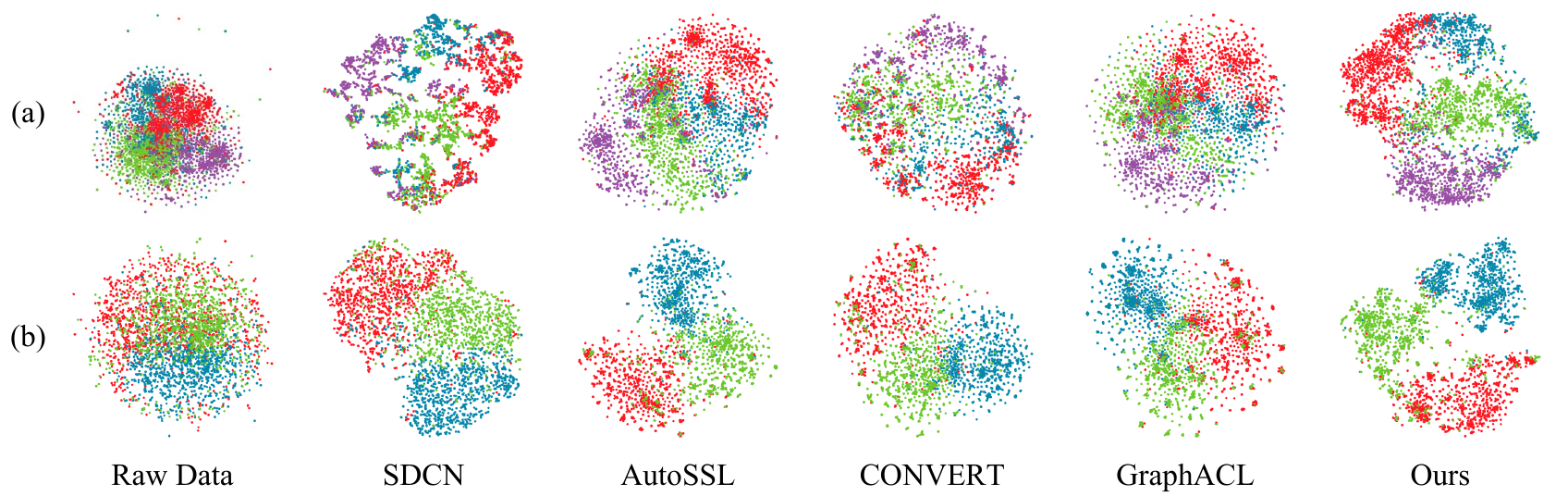}
        \captionsetup{justification=centering}
        \caption{t-SNE visualization comparison on (a) DBLP and (b) ACM datasets.}
    \label{fig_tsne}
    }
    \vspace{-1.0em}
\end{figure*}

\subsection{Visualization Analysis}
To evaluate the inherent clustering structure, we employ visualization to depict the distribution of the learned embeddings. Specifically, we compare our NeuCGC with other baselines in embedding learning using the t-SNE \cite{2008tsne-van} algorithm on the DBLP and ACM datasets. As shown in Fig. \ref{fig_tsne}, the visual results indicate that our method exhibits a more enhanced clustering structure compared to other methods.

\subsection{Parameter Analysis}
In this section, we conduct a comprehensive sensitivity analysis of key hyper-parameters. First, we jointly examine the two loss weighting factors $\lambda_1$ and $\lambda_2$, exploring values from $\{0.01, 0.1, 0.5, 1.0, 5, 10, 100\}$. Fig. \ref{fig_lambdas} presents the clustering accuracy across four datasets using different parameter combinations through 3D bar charts. Additionally, we investigate the effects of the high-confidence threshold $k$ (ranging from 0.1 to 1.0) and latent embedding dimension $d$ (varying from 100 to 2000), with results visualized in Fig. \ref{fig_k_hd}. From these figures, we derive several detailed observations below:
\begin{figure}[!t]
    \centering{
        \includegraphics[width=0.48\textwidth]{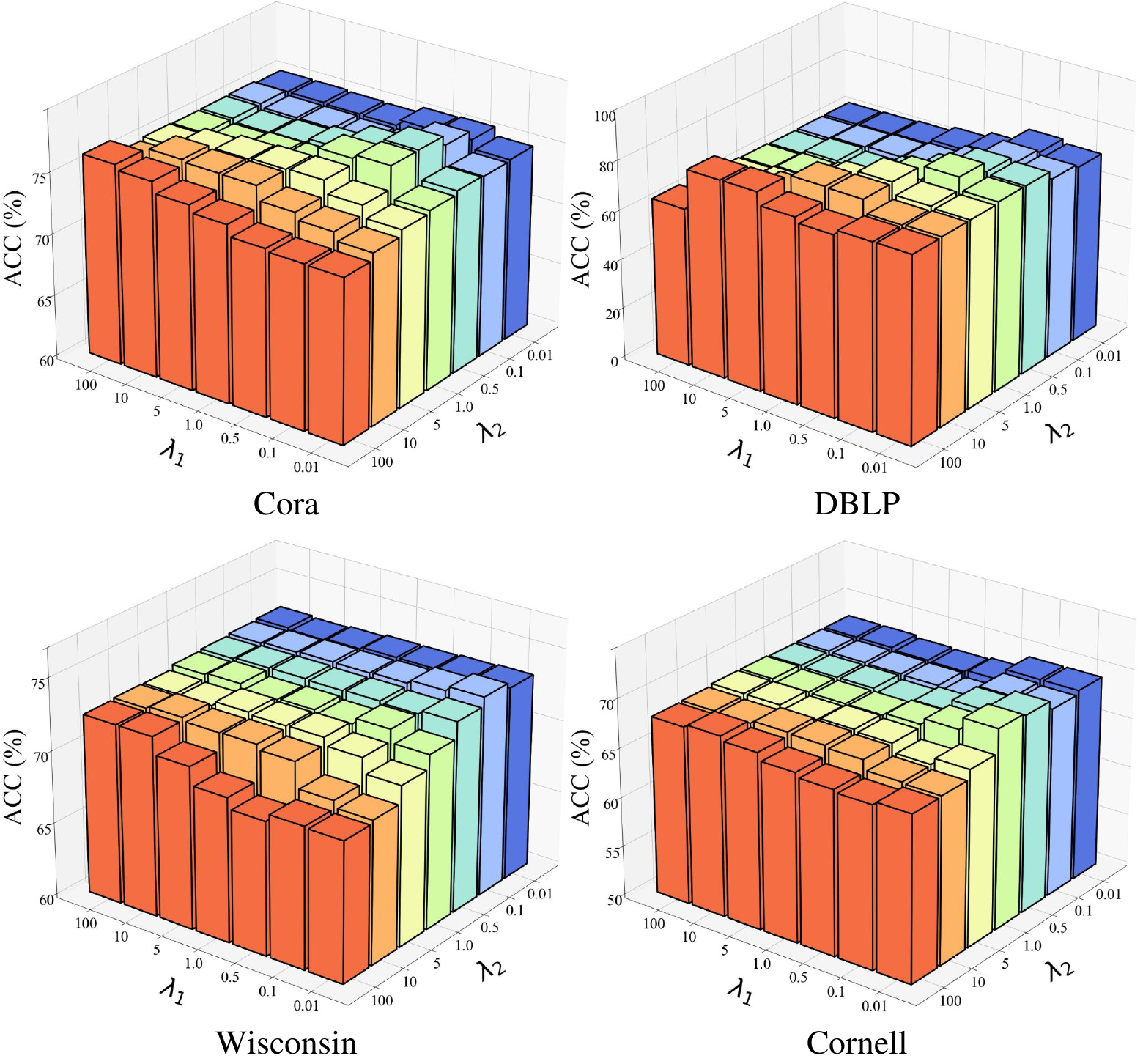}
        \captionsetup{justification=centering}
        \caption{Sensitivity analysis of loss weighting factors $\lambda_1$ and $\lambda_2$.}
    \label{fig_lambdas}
    }
    \vspace{-0.6em}
\end{figure}
\begin{figure}[!t]
    \centering{
        \includegraphics[width=0.48\textwidth]{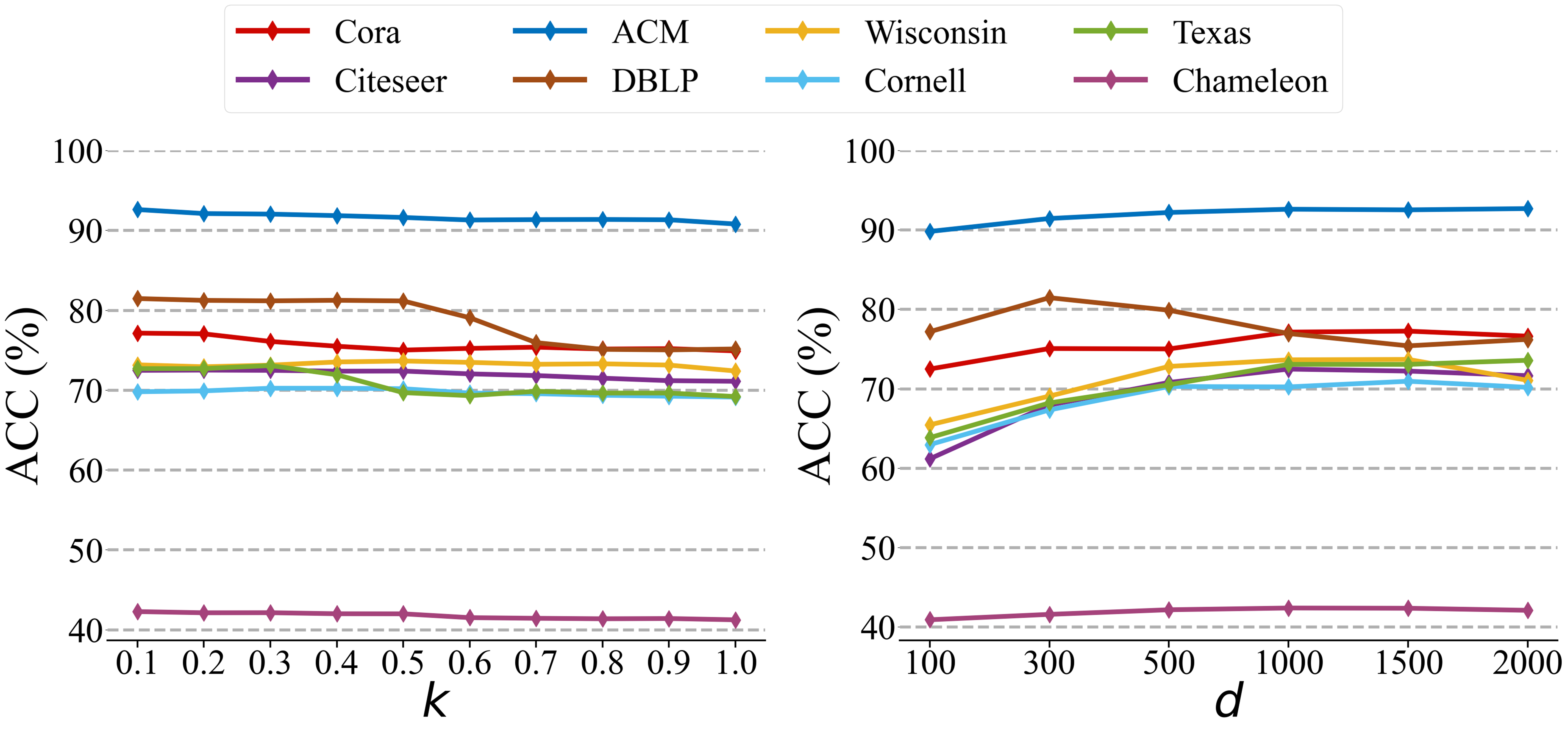}
        \captionsetup{justification=centering}
        \caption{Sensitivity analysis of the confidence signal $\textit{k}$ and dimension $\textit{d}$ of the latent embeddings.}
    \label{fig_k_hd}
    }
    \vspace{-0.8em}
\end{figure}
\begin{itemize}
    \item The model typically achieves high performance when $\lambda_1$ and $\lambda_2$ are relatively close, e.g., $\lambda_2$ is 1-10 times larger than $\lambda_1$. For instance, it attains optimal accuracy on Cora and DBLP at ($\lambda_1=0.1$, $\lambda_2=1.0$) and on Wisconsin at ($\lambda_1=0.01$, $\lambda_2=0.1$). However, when the gap between them is large, such as ($\lambda_1=0.01$, $\lambda_2=100$) or ($\lambda_1=100$, $\lambda_2=0.01$), the clustering accuracy significantly decreases across all four datasets. This suggests that appropriately balanced values of $\lambda_1$ and $\lambda_2$ promote the model to learn more discriminative features, leading to superior cluster assignment.
    \item The model exhibits strong robustness when the high-confidence threshold $k$ varies within the range $[0.1, 0.4]$ across all datasets. However, as $k$ increases beyond 0.4 and approaches 1.0, a noticeable decline in clustering accuracy is observed on Cora, DBLP, and Texas datasets. This degradation may be attributed to confirmation bias from over-confident pseudo-labels \cite{2020pseudo-label-Arazo}, which can adversely affect the model’s representation learning and ultimately degrade clustering performance.
    \item For embedding dimensions, insufficient $d$ (e.g., $<300$) may cause feature loss, thereby reducing representation quality. Whereas, performance typically improves as it increases on most datasets, with the trend converging around $d=1000$. However, an exception is observed on the DBLP dataset, where the clustering accuracy peaks at $d=300$ (matching its initial dimension 334) before declining. This suggests higher dimensions may introduce noise that interfere with the clustering process. Based on these findings, we set default dimension $d$ to 500 or 1000 for most datasets, while adopting $d=300$ for DBLP. 
\end{itemize} 

\subsection{Convergence Analysis}
Here, we empirically analyze the convergence behavior of our approach. As illustrated in Fig. \ref{fig_training_curves}, we present the loss values and several clustering performance metrics over the course of training on ACM and DBLP datasets, along with the number of iterations. The results reveal that the loss values decrease monotonically as the number of iterations increases. Notably, the values drop sharply within the first 200 iterations, which aligns with a significant improvement in clustering metric scores, before stabilizing in later iterations. These findings demonstrate that our method achieves rapid convergence and strong stability during training, ensuring superior clustering performance.
\begin{figure}[!t]
    \centering{
        \includegraphics[width=0.48\textwidth]{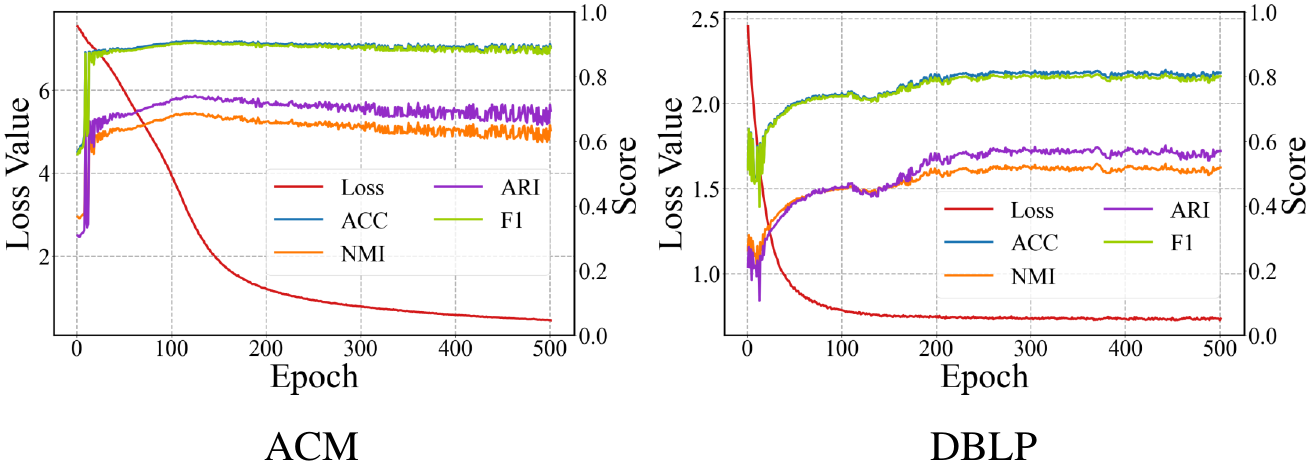}
        \captionsetup{justification=centering}
        \caption{Training curves on ACM and DBLP datasets.}
    \label{fig_training_curves}
    }
    \vspace{-0.5em}
\end{figure}

\section{Conclusion}
In this study, we propose a homophily-aware neighborhood contrastive learning framework NeuCGC for graph clustering. Unlike existing approaches that uniformly treat all neighboring nodes as positive pairs, our method introduces neutral pairs, which are adaptively weighted as positive pairs based on the graph's homophily level through a learned neighborhood contrastive factor. Leveraging this strategy, we develop neighborhood distribution neutral contrastive alignment and adaptive feature consistency neutral contrastive learning mechanisms. Experimental results demonstrate that our approach significantly outperforms state-of-the-art graph clustering methods in both performance and robustness across various graph settings.

\section*{Acknowledgment}
This work was supported in part by the GuangDong Basic and Applied Basic Research Foundation (Project No.2025A1515011692, 2023A1515030154, 2024A1515011437), in part by the Fundamental Research Funds for the Central Universities (Project No. 2024ZYGXZR077), in part by TCL Science and Technology Innovation Fund (Project No. 20231752), in part by the Research Grants Council of the Hong Kong Special Administration Region (Projection No. CityU 11206622).


\bibliographystyle{IEEEtran}
\bibliography{reference}


 





\end{document}